\documentclass[a4paper,10pt]{article}

\newcommand{\typeof}{1} %

\usepackage{ifthen}
\newcommand{\longv}[1]{\ifthenelse{\equal{\typeof}{1}}{#1}{}}
\newcommand{\shortv}[1]{\ifthenelse{\equal{\typeof}{0}}{#1}{}}

\usepackage{mathrsfs}

\longv{
\usepackage{a4wide}
}

\shortv{
\title{Complexity Analysis\\in Presence of Control Operators\\and Higher-Order Functions\thanks{An extended version of this paper including more details is available~\cite{EV}.}}
}

\longv{
\title{Complexity Analysis\\in Presence of Control Operators\\and Higher-Order Functions}
}

\shortv{
\author{Ugo Dal Lago \and Giulio Pellitta}
\institute{Universit\`a di Bologna \& INRIA Sophia Antipolis\\ \email{\{dallago,pellitta\}@cs.unibo.it}}
}

\longv{
\author{Ugo Dal Lago \and Giulio Pellitta}
\date{}
}

\newcommand{\condinc}[2]{#1}

\newcommand{\BLLPtoLLP}[1]{\langle #1\rangle}
\newcommand{\BLLPtoLLPpn}[1]{\langle\!\langle #1\rangle\!\rangle}

\usepackage{cmll,amsmath,bussproofs,amssymb,enumerate,color}
\usepackage[safe]{tipa}
\condinc{}{\usepackage{a4wide}}
\condinc{}{\usepackage{multicol,hyperref}}
\condinc{}{\usepackage{tikz}

}
\condinc{\usepackage[pdfpagelabels,plainpages=false]{hyperref}}{}
\usepackage{pdflscape}

\newcommand{\remove}[1]{}

\newcommand{\SUM}{\sum}

\newcommand{\WITH}{\with}
\newcommand{\TENSOR}{\otimes}
\newcommand{\PLUS}{\oplus}
\newcommand{\PAR}{\parr}

\newcommand{\AX}{\mathsf{Ax}}
\newcommand{\CUT}{\mathsf{Cut}}
\newcommand{\NOT}[1]{{#1}^\perp}
\newcommand{\SUBTYPE}{\sqsubseteq}
\newcommand{\SUPTYPE}{\sqsupseteq}

\newcommand{\AND}{\wedge}

\newcommand{\nfone}{N}
\newcommand{\nftwo}{M}
\newcommand{\nfthree}{O}
\newcommand{\nffour}{K}

\newcommand{\pfone}{P}
\newcommand{\pftwo}{Q}
\newcommand{\pfthree}{R}
\newcommand{\pffour}{S}
\newcommand{\contextone}{\Gamma}
\newcommand{\contexttwo}{\Delta}
\newcommand{\contextthree}{\Theta}
\newcommand{\contextfour}{\Xi}
\newcommand{\contextfive}{\Psi}
\newcommand{\contextsix}{\Phi}
\newcommand{\contextseven}{\Upsilon}
\newcommand{\contexteight}{\Pi}
\newcommand{\negcontextone}{\mathcal{N}}
\newcommand{\negcontexttwo}{\mathcal{M}}
\newcommand{\negcontextthree}{\mathcal{O}}
\newcommand{\negcontextfour}{\mathcal{K}}

\newcommand{\proofone}{\pi}
\newcommand{\prooftwo}{\rho}
\newcommand{\proofthree}{\sigma}
\newcommand{\prooffour}{\lambda}

\newcommand{\pfd}[1]{#1^\diamond}

\newcommand{\polone}{p}
\newcommand{\poltwo}{q}
\newcommand{\polthree}{r}
\newcommand{\polfour}{s}
\newcommand{\polfive}{t}

\newcommand{\polnine}{h}
\newcommand{\polten}{k}
\newcommand{\poleleven}{m}

\newcommand{\resleq}{\sqsubseteq}
\newcommand{\reslt}{\sqsubset}
\newcommand{\resgt}{\sqsupset}
\newcommand{\resgeq}{\sqsupseteq}

\newcommand{\varone}{x}
\newcommand{\vartwo}{y}
\newcommand{\varthree}{z}
\newcommand{\varfour}{u}
\newcommand{\varfive}{v}
\newcommand{\varsix}{c}
\newcommand{\varseven}{d}
\newcommand{\vareight}{a}
\newcommand{\varnine}{b}
\newcommand{\varten}{f}

\newcommand{\lvarone}{x}

\newcommand{\mvarone}{\alpha}

\newcommand{\typeone}{A}
\newcommand{\typetwo}{B}

\newcommand{\pltone}{\mathbf{P}}
\newcommand{\plttwo}{\mathbf{Q}}
\newcommand{\pltthree}{\mathbf{R}}

\newcommand{\nltone}{\mathbf{N}}
\newcommand{\nlttwo}{\mathbf{M}}
\newcommand{\nltthree}{\mathbf{L}}
\newcommand{\nltfour}{\mathbf{O}}

\newcommand{\ltone}{\mathbf{A}}
\newcommand{\lttwo}{\mathbf{B}}

\newcommand{\ATOMONE}{V}
\newcommand{\ATOMTWO}{U}
\newcommand{\COATOMONE}{X}
\newcommand{\COATOMTWO}{Y}
\newcommand{\LET}[3]{#1\{#2:=#3\}}
\newcommand{\subst}[3]{#1\{#2/#3\}}

\newcommand{\DEF}{:=}
\newcommand{\contsum}{\uplus}

\newcommand{\intone}{n}

\newcommand{\intthree}{k}
\newcommand{\intfour}{l}
\newcommand{\prodvar}{m}
\newcommand{\idxone}{i}
\newcommand{\idxtwo}{j}
\newcommand{\muvarone}{\alpha}
\newcommand{\muvartwo}{\beta}
\newcommand{\muvarthree}{\gamma}
\newcommand{\muvarfour}{\delta}
\newcommand{\termone}{t}
\newcommand{\termtwo}{u}
\newcommand{\termthree}{v}
\newcommand{\termfour}{w}

\newcommand{\K}{\mathsf{K}}
\newcommand{\PCF}{\mathsf{PCF}}
\newcommand{\muPCF}{\mu\mathsf{PCF}}
\newcommand{\BLL}{\mathsf{BLL}}
\newcommand{\LLP}{\mathsf{LLP}}
\newcommand{\BLLP}{\mathsf{BLLP}}
\newcommand{\BLLPLM}{\mathsf{BLLP}_{\lambda\mu}}
\newcommand{\BLLPLMM}{\mathsf{BLLP}_{\lambda\mu}^\mathsf{mult}}
\newcommand{\LABEL}[3]{[#1]_{#2}^{#3}}

\newcommand{\DLABEL}[5]{\mbox{${}_{#1}^{#2}{[{#3}]}_{#4}^{#5}$}}
\newcommand{\arrow}[2]{\multimap_{#1}^{#2}}
\newcommand{\any}{*}
\newcommand{\pneg}[3]{\neg_{#1}^{#2} {#3}}

\newcommand{\sumlf}[3]{\sum_{#1<#2}#3}
\newcommand{\streq}{\sim}
\newcommand{\pof}{\;\triangleright\;}

\newcommand{\seq}{\vdash}

\newcommand{\callcc}{\texttt{callcc}}
\newcommand{\FellC}{\ensuremath{\mathcal{C}}}

\newcommand{\Scheme}{\ensuremath{\mathsf{Scheme}}}

\newcommand{\grem}[2][unknown]{\textcolor{blue}{\textbf{Giulio} [data #1]: #2}}

\newcommand{\srvone}{S}
\newcommand{\srvp}[2]{S^{#1}(#2)}

\newcommand{\cutelimarr}{\longrightarrow}
\newcommand{\redcutelimarr}{\Longrightarrow}

\newcommand{\logred}{\longmapsto}
\newcommand{\commeq}{\cong}

\newcommand{\tj}[4]{#1\;\vdash\;#2:#3\;|\;#4}
\newcommand{\tjp}[5]{#1\;\vdash_{#2}\;#3:#4\;|\;#5}

\newcommand{\envone}{\mathscr{E}}
\newcommand{\envtwo}{\mathscr{F}}

\newcommand{\stone}{\mathscr{S}}
\newcommand{\sttwo}{\mathscr{T}}

\newcommand{\closone}{\mathcal{C}}
\newcommand{\clostwo}{\mathcal{D}}

\newcommand{\confone}{C}
\newcommand{\conftwo}{D}

\newcommand{\toconf}{\hookrightarrow}

\newcommand{\tjc}[2]{\vdash #1:#2}

\newcommand{\pair}[2]{\langle#1,#2\rangle}
\newcommand{\emptyenv}{\emptyset}
\newcommand{\emptystk}{\varepsilon}

\condinc{}
{
\newtheorem*{Def*}{Definition}

\newtheorem*{Pro*}{Proposition}

\newtheorem*{Thm*}{Theorem}

\newtheorem*{Rem*}{Remark}

\newtheorem*{Lem*}{Lemma}

\newtheorem*{Exa*}{Example}
\newtheorem*{Not*}{Notation}

\newtheorem*{Cor*}{Corollary}
}

\longv{
\newenvironment{proof}{\begin{trivlist}
       \item[\hskip \labelsep {\bfseries Proof.}]}{\hfill $\Box$ \end{trivlist}}
}

\newcommand{\FV}[1]{\mathit{FV}(#1)}

\newenvironment{varitemize}
{
\begin{list}{\shortv{\labelitemii}\longv{\labelitemi}}
{
\setlength{\itemsep}{0pt}
 \setlength{\topsep}{0pt}
 \setlength{\parsep}{0pt}
 \setlength{\partopsep}{0pt}
 \setlength{\leftmargin}{15pt}
 \setlength{\rightmargin}{0pt}
 \setlength{\itemindent}{0pt}
 \setlength{\labelsep}{5pt}
 \setlength{\labelwidth}{10pt}}}
{
 \end{list}
}

\newcommand{\midd}{\; \; \mbox{\Large{$\mid$}}\;\;}

\longv{
\newtheorem{theorem}{Theorem}[section]
\newtheorem{lemma}[theorem]{Lemma}
\newtheorem{proposition}[theorem]{Proposition}

}

\begin{document}

\linespread{0.915}

\maketitle

\begin{abstract}
A polarized version of Girard, Scedrov and Scott's Bounded Linear Logic is introduced and its
normalization properties studied. Following Laurent~\cite{laurent2003polarized}, the logic naturally
gives rise to a type system for the $\lambda\mu$-calculus, 
whose derivations reveal bounds on the time complexity of the underlying term. This is the first example of a type system for 
the $\lambda\mu$-calculus guaranteeing time complexity bounds for typable programs.
\end{abstract}

\section{Introduction}
Among non-functional properties of programs, bounds on the amount of resources 
(like computation time and space) programs need when executed are particularly significant. The problem of deriving such
bounds is indeed crucial in safety-critical systems, but is undecidable whenever non-trivial
programming languages are considered. If the units of measurement become concrete and close to the physical
ones, the problem becomes even more complicated and architecture-dependent. A typical example is the one of WCET techniques adopted
in real-time systems~\cite{wcet}, which \remove{do }not only need to deal with how many machine instructions a program corresponds
to, but also with how much time each instruction costs when executed by possibly complex architectures
(including caches, pipelining, etc.), a task which is becoming even harder with the current trend
towards multicore architectures.

A different approach consists in analysing the \emph{abstract} complexity of programs. As an example,
one can take the number of instructions executed by the program as a measure of its execution
time. This is of course a less informative metric, which however becomes more accurate if the actual 
time taken \emph{by each instruction} is kept low. One advantage of this analysis is the independence 
from the specific hardware platform executing the program at hand: the latter only needs
to be analysed once. A variety of \emph{complexity analysis} techniques have been employed in this context, from
abstract interpretation~\cite{speed} to type systems~\cite{HOAA10} to program logics~\cite{deBakker80} 
to interactive theorem proving. Properties of programs written in higher-order functional languages are for 
various reasons well-suited to be verified by way of type systems. This includes not only safety properties 
(e.g. well-typed programs do not go wrong), but more complex ones, including resource bounds~\cite{HOAA10,BaillotTeruiIC,GaboardiRonchi,DalLago2011}. 

In this paper, we delineate a methodology for complexity analysis of higher-order programs \emph{with control
operators}. The latter are constructs which are available in most concrete functional programming languages (including
\textsf{Scheme} and \textsf{OCaml}), and allow control to flow in non-standard ways. The technique 
we introduce takes the form of a type system for de Groote's $\lambda\mu$-calculus~\cite{de1994relation} derived 
from Girard, Scedrov and Scott's Bounded Linear Logic~\cite{GSS92TCS} ($\BLL$ in the following). 
We prove it to be sound: typable programs can indeed be reduced in a number of steps lesser or equal to 
a (polynomial) bound which can be read from the underlying type derivation.
A similar result can be given when the cost model is the one induced by an abstract machine. To the authors'
knowledge, this is the first example of a complexity analysis methodology coping well not only with higher-order
functions, but also with control operators.

In the rest of this section, we explain the crucial role Linear Logic has in this work, in the meantime delineating
its main features.
\subsection{Linear Logic and Complexity Analysis}
Linear Logic~\cite{Girard87} is one of the most successful tools for characterizing complexity classes in
a higher-order setting, through the Curry-Howard correspondence. Subsystems of it can indeed be shown to correspond to the polynomial 
time computable functions~\cite{GSS92TCS,Girard98IC,Lafont04TCS} or the logarithmic 
space computable functions~\cite{Schoepp07LICS}. Many of the introduced fragments can then be turned into
type systems for the $\lambda$-calculus~\cite{BaillotTeruiIC,GaboardiRonchi}, some of them being
relatively complete in an intensional sense~\cite{DalLago2011}.

The reason for this success lies in the way Linear Logic decomposes intuitionistic implication into a linear implication,
which has low complexity, and an \emph{exponential modality}, which marks those formulas to which structural rules
can be applied. This gives a proper status to proof duplication, without which cut-elimination can be 
performed in a linear number of steps. By tuning the rules governing the exponential modality, then, one can define
logical systems for which cut-elimination can be performed within appropriate resource bounds. Usually, this is
coupled with an encoding of all functions in a complexity class $\mathcal{C}$ into the system at hand, which makes
the system a \emph{characterization} of $\mathcal{C}$.

Rules governing the exponential modality $!$ can be constrained in (at least) two different 
ways:\begin{varitemize}
\item
  On the one hand, one or more of the rules governing $!$ (e.g., dereliction or digging) 
  can be \emph{dropped} or \emph{restricted}. This is what happens, for example, in Light Linear Logic~\cite{Girard98IC} or Soft
  Linear Logic~\cite{Lafont04TCS}.
\item
  On the other, the logic can be further refined and \emph{enriched} so as to control the number of times structural
  rules are applied. In other words, rules for $!$ are still all there, but in a refined form. This is
  what happens in Bounded Linear Logic~\cite{GSS92TCS}. Similarly, one could control so-called modal impredicativity
  by a system of levels~\cite{BaillotMazza10}.
\end{varitemize} 
The first approach corresponds to cutting the space of proofs with an axe: many proofs, and among
them many corresponding to efficient algorithms, will not be part of the system because they require one of
the forbidden logical principles. The second approach is milder in terms of the class of good programs that are ``left
behind'': there is strong evidence that with this approach one can obtain a quite expressive logical system~\cite{DalLago2009,DalLago2011}.

Not much is known about whether this approach scales to languages in which not only functions but also
first-class continuations and control operators are present. Understanding the impact of these features to the complexity
of programs is an interesting research topic, which however has received little attention in the past.
\subsection{Linear Logic and Control Operators}
On the other hand, more than twenty years have passed since classical logic has been shown to be amenable to the
Curry-Howard paradigm~\cite{Griffin90}. And, interestingly enough, classical axioms (e.g. Pierce's law or the law of
the Excluded Middle) can be seen as the type of control operators like \textsf{Scheme}'s \texttt{callcc}. In the meantime, the various facets of this
new form of proofs-as-programs correspondence have been investigated in detail, and many extensions of the $\lambda$-calculus
for which classical logic naturally provides a type discipline have been introduced (e.g. \cite{Parigot,CurienHerbelin}).

Moreover, the decomposition provided by Linear Logic is known to scale up to classical logic~\cite{Girard91}. Actually, Linear Logic
was known to admit an involutive notion of negation from its very inception~\cite{Girard87}. A satisfying embedding of Classical
Logic into Linear Logic, however, requires restricting the latter by way of polarities~\cite{phdlaurent}: this way one is left with
a logical system with most of the desirable dynamical properties.

In this paper, we define $\BLLP$, a polarized version of Bounded Linear Logic. The kind of
enrichment resource polynomials provide in $\BLL$ is shown to cope well with polarization.
Following the close relationship between polarized linear logic and the $\lambda\mu$-calculus~\cite{laurent2003polarized}, 
$\BLLP$ gives rise to a type system for the $\lambda\mu$-calculus. Proofs and typable $\lambda\mu$-terms 
are both shown to be reducible to their cut-free or normal forms in a number of steps bounded by a polynomial
weight. Such a result for the former translates to a similar result for the latter, since any reduction step in 
$\lambda\mu$-terms corresponds to one or more reduction steps in proofs. The analysis is then extended to the reduction of 
$\lambda\mu$-terms by a Krivine-style abstract machine~\cite{de1998environment}.
\section{Bounded Polarized Linear Logic as A Sequent Calculus}
In this section, we define $\BLLP$ as a sequent calculus. Although this section is self-contained,
some familiarity with both bounded~\cite{GSS92TCS} and polarized~\cite{laurent2003polarized} linear logic would
certainly help. \shortv{Some more details can be found in an extended version of the
present paper~\cite{EV}.}
\subsection{Polynomials and Formulas}\label{sect:polform}
A \emph{resource monomial} is any (finite) product of binomial coefficients in the form
$\displaystyle\prod_{\idxone=1}^{\prodvar} \left({\varone_\idxone \atop \intone_\idxone}\right)$, 
where the $\varone_\idxone$ are distinct variables and the $\intone_\idxone$ are non-negative integers.
A \emph{resource polynomial} is any finite sum of resource monomials.
Given resource polynomials $\polone, \poltwo$ write $\polone\resleq \poltwo$ to denote that 
$\poltwo-\polone$ is a resource polynomial. If $\polone\resleq \polthree$ and $\poltwo\resleq \polfour$ 
then also $\poltwo\circ \polone \resleq \polfour\circ\polthree$. Resource polynomials are closed
by addition, multiplication, bounded sums and composition~\cite{GSS92TCS}. 

A \emph{polarized formula} is a formula (either positive or negative) generated by the following grammar
\shortv{
\begin{align*}
  \pfone&::=\ATOMONE\midd\pfone\TENSOR\pfone\midd 1
    \midd !_{\varone<\polone}\nfone;\\
  \nfone&::=\NOT{\ATOMONE}\midd\nfone\PAR\nfone\midd \bot 
    \midd ?_{\varone<\polone}\pfone;
\end{align*}}
\longv{
\begin{align*}
  \pfone&::=\ATOMONE(\vec \polone) \midd \pfone\TENSOR \pfone \midd \remove{\pfone\PLUS \pfone \midd} 1 
  \midd \exists \ATOMONE \pfone \midd !_{\varone<\polone}\nfone;\\
  \nfone&::=\NOT{\ATOMONE}(\vec \polone) \midd \nfone\PAR\nfone \midd \remove{\nfone\WITH\nfone \midd} \bot
  \midd \forall \ATOMONE \nfone \midd ?_{\varone<\polone}\pfone.
\end{align*}
}where $\ATOMONE$ ranges over a countable sets of atoms. 
Throughout this paper, formulas (but also terms, contexts, etc.) are considered modulo $\alpha$-equivalence.
Formulas (either positive or negative) are ranged over by metavariables like $\typeone,\typetwo$. Formulas
like $\NOT{\ATOMONE}$ are sometime denoted as $\COATOMONE,\COATOMTWO$.

In a polarized setting, contraction can be performed on any negative formula. As a consequence, we need the notion
of a \emph{labelled\footnote{Keep in mind that linear logic contains 
a subset of formulas which is isomorphic to (polarized) classical logic. $\LABEL{\nfone}{\varone}{\polone}$ 
(resp. $\LABEL{\pfone}{\varone}{\polone}$) can be thought of roughly as $?_{\varone < \polone}\NOT{\nfone}$ 
(resp. $!_{\varone < \polone}\NOT{\pfone}$), i.e., in a sense we can think of labelled formulas 
as formulas hiding an implicit exponential modality.} formula} $\LABEL{\typeone}{\varone}{\polone}$, 
namely the \emph{labelling} of the formula $\typeone$ with respect to $\varone$ and $\polone$. 
All occurrences of $\varone$ in $\typeone$ are bound in 
$\LABEL{\typeone}{\varone}{\polone}$. Metavariables for labellings of positive (respectively, negative) 
formulas are $\pltone,\plttwo,\pltthree$ (respectively, $\nltone,\nlttwo,\nltthree$). Labelled formulas are sometimes
denoted with metavariables $\ltone,\lttwo$ when their polarity is not essential. Negation, as usual in classical linear
system, can be applied to any (possibly labelled) formula, \emph{\`a la} De Morgan. When the resource
variable $\varone$ does not appear in $\typeone$, then we do not need to mention it when writing
$\LABEL{\typeone}{\varone}{\polone}$, which becomes $\LABEL{\typeone}{}{\polone}$. Similarly for
$!_{\varone<\polone}\nfone$ and $?_{\varone<\polone}\pfone$. 

Both the space of formulas and the space of
labelled formulas can be seen as partial orders by stipulating that two (labelled) 
formulas can be compared \remove{only if}iff they have \emph{exactly} the
same skeleton and the polynomials occurring in them can be compared.
\shortv{
As an example,
\begin{align*}
  !_{\varone<\polone}\nfone\SUBTYPE\; !_{\varone<\poltwo}\nftwo &\mbox{ iff } \poltwo\resleq \polone \AND \nfone \SUBTYPE \nftwo;\\
  ?_{\varone<\polone}\pfone\SUBTYPE\; ?_{\varone<\poltwo}\pftwo &\mbox{ iff } \polone\resleq \poltwo \AND \pfone\SUBTYPE \pftwo.
\end{align*}}\longv{
Formally,
\begin{align*}
  \ATOMONE(\polone_1, \dots, \polone_\intone)&\SUBTYPE \ATOMONE(\poltwo_1, \dots, \poltwo_\intone) \mbox{ iff } 
    \forall \idxone.\polone_\idxone\resleq \poltwo_\idxone;\\
  \NOT{\ATOMONE}(\polone_1, \dots, \polone_\intone)&\SUBTYPE \NOT{\ATOMONE}(\poltwo_1, \dots, \poltwo_\intone) \mbox{ iff } 
    \forall \idxone.\poltwo_\idxone\resleq \polone_\idxone;\\
  1 &\SUBTYPE 1;\\
  \bot &\SUBTYPE \bot;\\
  \pfone \TENSOR \pftwo &\SUBTYPE \pfthree \TENSOR \pffour \mbox{ if{}f }\pfone \SUBTYPE \pfthree \AND \pftwo \SUBTYPE \pffour;\\
  \nfone \PAR \nftwo &\SUBTYPE \nfthree \PAR \nffour \mbox{ if{}f } \nfone \SUBTYPE \nfthree \AND \nftwo \SUBTYPE \nffour;\\
  !_{\varone<\polone}\nfone&\SUBTYPE !_{\varone<\poltwo}\nftwo \mbox{ if{}f } \poltwo\resleq \polone \AND \nfone \SUBTYPE \nftwo;\\
  ?_{\varone<\polone}\pfone&\SUBTYPE ?_{\varone<\poltwo}\pftwo \mbox{ if{}f } \polone\resleq \poltwo \AND \pfone\SUBTYPE \pftwo;\\
  \forall \ATOMONE.\nfone&\SUBTYPE \forall \ATOMONE.\nftwo \mbox{ if{}f } \nfone \SUBTYPE \nftwo;\\
  \exists \ATOMONE.\pfone&\SUBTYPE \exists \ATOMONE.\pftwo \mbox{ if{}f } \pfone \SUBTYPE \pftwo.
\end{align*}
}
In a sense, then, polynomials occurring next to atoms or
to the \emph{whynot} operator are in positive position, while those 
occurring next to the \emph{bang} operator are in negative position. In all the other
cases, $\resleq$ is defined component-wise, in the natural way, e.g. $\pfone\TENSOR\pftwo\resleq\pfthree\TENSOR\pffour$
iff both $\pfone\resleq\pfthree$ and $\pftwo\resleq\pffour$. Finally $\LABEL{\nfone}{\varone}{\polone}\SUBTYPE \LABEL{\nftwo}{\varone}{\poltwo}$ 
iff $\nfone \SUBTYPE \nftwo \AND \polone \resgeq \poltwo$. And dually, $\LABEL{\pfone}{\varone}{\polone}\SUBTYPE \LABEL{\pftwo}{\varone}{\poltwo}$ 
iff $\nfone \SUBTYPE \nftwo \AND \polone \resleq \poltwo$.
\longv{
  \begin{lemma}\label{Lem:subneg}
    $\typeone\SUBTYPE\typetwo$ iff $\NOT{\typetwo}\SUBTYPE\NOT{\typeone}$.
    Moreover, $\ltone\SUBTYPE\lttwo$ iff $\NOT{\lttwo}\SUBTYPE\NOT{\ltone}$. 
  \end{lemma}
  \begin{proof}
    $\typeone\SUBTYPE\typetwo$ iff $\NOT{\typetwo}\SUBTYPE\NOT{\typeone}$ can be proved by induction
    on the structure of $\typeone$. Consider the second part of the statement, now.
    Suppose that $\typeone, \typetwo$ are positive, and call them $\pfone, \pftwo$ respectively. Then
    \begin{align*} 
      \LABEL{\pfone}{\varone}{\polone} \SUBTYPE \LABEL{\pftwo}{\varone}{\poltwo}&\Leftrightarrow\pfone \SUBTYPE \pftwo\wedge\polone \resleq \poltwo\\
        &\Leftrightarrow\NOT{\pftwo}\SUBTYPE\NOT{\pfone}\wedge\polone \resleq \poltwo\\
        &\Leftrightarrow\LABEL{\NOT{\pftwo}}{\varone}{\poltwo}\SUBTYPE\LABEL{\NOT{\pfone}}{\varone}{\polone}.
    \end{align*}
    The case when $\typeone,\typetwo$ are negative is similar.
  \end{proof}
}
  
Certain operators on resource polynomials can be lifted to formulas. 
As an example, we want to be able to \emph{sum} labelled formulas provided they have a 
proper form:
$$
\LABEL{\nfone}{\varone}{\polone}\contsum\LABEL{\subst{\nfone}{\varone}{\vartwo+\polone}}{\vartwo}{\poltwo}\DEF
\LABEL{\nfone}{\varone}{\polone+\poltwo}.
$$
We are assuming, of course, that $\varone,\vartwo$ are not free in either $\polone$ or $\poltwo$.
This construction can be generalized to \emph{bounded} sums:
suppose that a labelled formula is in the form
$$
  \LABEL{\nftwo}{\vartwo}{\polthree}=\LABEL{\nfone\{\varone/\vartwo+\sum_{\varfour<\varthree}
    \subst{\polthree}{\varthree}{\varfour}\}}{\vartwo}{\polthree},
$$
where $\vartwo$ and $\varfour$ are not free in $\nfone$ nor in $\polthree$ and $\varthree$ 
is not free in $\nfone$. Then the labelled formula $\sum_{\varthree<\poltwo}\LABEL{\nftwo}{\vartwo}{\polthree}$
is defined as $\LABEL{\nfone}{\varone}{\sum_{\varthree<\poltwo}\polthree}$.
See \cite[\S 3.3]{GSS92TCS} for more details about the above constructions.

\longv{
  An \emph{abstraction formula} of arity $n$ is simply a formula $\typeone$, where
  the $n$ resource variables $\varone_1,\ldots,\varone_\intone$ are meant to be bound.
  $\LET{\typeone}{\ATOMONE}{\typetwo}$ is the result of 
  substituting a second order abstraction term $\typetwo$ (of arity $\intone$) for all 
  free occurrences of the propositional variable $\ATOMONE$ (of the same arity) in $\typeone$. This can be defined formally
  by induction on the structure of $\typeone$, but the only interesting clauses are the following two:
  \begin{align*}
    \LET{\ATOMONE(\polone_1,\dots,\polone_\intone)}{\ATOMONE}{\typetwo}&=\subst{\typetwo}{\varone_1,\ldots,\varone}{\polone_1,\dots,\polone_\intone}\\
    \LET{\NOT{\ATOMONE}(\polone_1,\dots,\polone_\intone)}{\ATOMONE}{\typetwo}&=\subst{\NOT{\typetwo}}{\varone_1,\ldots,\varone}{\polone_1,\dots,\polone_\intone}
  \end{align*}
}

\subsection{Sequents and Rules}\label{sect:sequents}
The easiest way to present $\BLLP$ is to give a sequent calculus for it. Actually, proofs will be
structurally identical to proofs of Laurent's $\LLP$. Of course, only \emph{some} of $\LLP$ proofs
are legal $\BLLP$ proofs --- those giving rise to an exponential blow-up cannot be decorated according
to the principles of Bounded Linear Logic.

A \emph{sequent} is an expression in the form $\seq\contextone$, where $\contextone=\ltone_1,\ldots
\ltone_n$ is a multiset of labelled formulas such
that at most one among $\ltone_1,\ldots,\ltone_n$ is positive. If $\contextone$ only
contains (labellings of) negative formulas, we indicate it with metavariables like $\negcontextone,\negcontexttwo$.
The operator $\contsum$ can be extended to one on multi-sets of formulas component-wise, so we can write
expressions like $\negcontextone \contsum \negcontexttwo$: this amounts to summing the polynomials
occurring in $\negcontextone$ and those occurring in $\negcontexttwo$. Similarly for bounded sums.

The rules of the sequent calculus for $\BLLP$ are in Figure~\ref{fig:sequentcalc}.
\begin{figure*}
\fbox{\shortv{\scriptsize}
\begin{minipage}{.98\textwidth}
\centering
  \ \vspace{5pt} \\
  \AxiomC{$\nltone \SUBTYPE \nlttwo$}
  \AxiomC{$\NOT{\nlttwo} \SUPTYPE \pltone$}
  \RightLabel{$\AX$}
  \BinaryInfC{$\seq \nltone, \pltone$}
  \DisplayProof
  \qquad
  \AxiomC{$\seq \contextone, \nltone$}
  \AxiomC{$\seq \negcontextone, \NOT{\nltone}$}
  \RightLabel{$\CUT$}
  \BinaryInfC{$\seq \contextone, \negcontextone$}
  \DisplayProof
  \\ \vspace{5pt}
  \AxiomC{$\seq \contextone, \LABEL{\nfone}{\varone}{\polone}, \LABEL{\nftwo}{\varone}{\poltwo}$}
  \AxiomC{$\polone\resleq\polthree$\quad$\poltwo\resleq\polthree$}
  \RightLabel{$\PAR$}
  \BinaryInfC{$\seq \contextone, \LABEL{\nfone\PAR\nftwo}{\varone}{\polthree}$}
  \DisplayProof
  \qquad
  \AxiomC{$\seq \negcontextone, \LABEL{\pfone}{\varone}{\polone}$}
  \AxiomC{$\seq \negcontexttwo, \LABEL{\pftwo}{\varone}{\poltwo}$\quad$\polthree\resleq\polone$\quad$\polthree\resleq\poltwo$}
  \RightLabel{$\TENSOR$}
  \BinaryInfC{$\seq \negcontextone, \negcontexttwo, \LABEL{\pfone\TENSOR \pftwo}{\varone}{\polthree}$}
  \DisplayProof
  \\ \vspace{5pt}
  \AxiomC{$\seq \negcontextone, \LABEL{\nfone}{\varone}{\polone}$}
  \AxiomC{$\negcontexttwo \SUBTYPE \SUM_{y<q}\negcontextone$}
  \RightLabel{!}
  \BinaryInfC{$\seq \negcontexttwo, \LABEL{!_{\varone<\polone}\nfone}{\vartwo}{\poltwo}$}
  \DisplayProof
  \qquad
  \AxiomC{$\seq \negcontextone, \LABEL{\subst{\pfone}{\vartwo}{0}}{\varone}{\subst{\polone}{y}{0}}$}
  \AxiomC{$\nltone\SUBTYPE\LABEL{?_{\varone<\polone}\pfone}{\vartwo}{1}$}
  \RightLabel{$?d$}
  \BinaryInfC{$\seq \negcontextone,\nltone$}
  \DisplayProof
  \\ \vspace{5pt}
  \AxiomC{$\seq \contextone$}
  \RightLabel{$?w$}
  \UnaryInfC{$\seq \contextone,\nltone$}
  \DisplayProof
  \qquad
  \AxiomC{$\seq \contextone,\nltone,\nlttwo$}
  \AxiomC{$\nltthree\SUBTYPE\nltone\contsum\nlttwo$}
  \RightLabel{$?c$}
  \BinaryInfC{$\seq \contextone,\nltthree$}
  \DisplayProof
  \qquad
  \AxiomC{$\seq \contextone$}
  \RightLabel{$\bot$}
  \UnaryInfC{$\seq \contextone, \LABEL{\bot}{\varone}{\polone}$}
  \DisplayProof
  \qquad
  \AxiomC{$\vphantom{\seq \contextone}$}
  \RightLabel{$1$}
  \UnaryInfC{$\seq \LABEL{1}{\varone}{\polone}$}
  \DisplayProof
  \shortv{
  \\ \ \vspace{5pt}}
  \longv{
  \\ \vspace{5pt}
  \AxiomC{$\seq \contextone, \LABEL{\nfone}{\varone}{\polone}$}
  \AxiomC{$\ATOMONE\not\in\FV{\nfone}$}
  \RightLabel{$\forall$}
  \BinaryInfC{$\seq \contextone, \LABEL{\forall \ATOMONE \nfone}{\varone}{\polone}$}
  \DisplayProof
  \qquad
  \AxiomC{$\seq\negcontextone, \LABEL{\LET{\pfone}{\ATOMONE}{\pftwo}}{\varone}{\polone}$}
  \RightLabel{$\exists$}
  \UnaryInfC{$\seq\negcontextone, \LABEL{\exists \ATOMONE \pfone}{\varone}{\polone}$}
  \DisplayProof\\
  \vspace{5pt} \ }\\
\end{minipage}}
\caption{$\BLLP$, Sequent Calculus Rules}\label{fig:sequentcalc}
\end{figure*}
Please observe that:
\begin{varitemize}
\item
  The relation $\SUBTYPE$ is implicitly applied to both
  formulas and polynomials whenever possible in such a way that ``smaller''
  formulas can always be derived (see Section \ref{sect:malleability}).
\item
  As in $\LLP$, structural rules can act on any negative formula, and
  not only on exponential ones. Since all formulas occurring in sequents are 
  labelled, however, we can still keep track of how many times formulas are
  ``used'', in the spirit of $\BLL$.
\item
  A byproduct of taking sequents as multisets of \emph{labeled} formulas 
  is that multiplicative rules themselves need to deal with labels. As 
  an example, consider rule $\TENSOR$: the resource polynomial labelling
  the conclusion $\pfone\TENSOR\pftwo$ is anything smaller or equal to the
  polynomials labeling the two premises.
\end{varitemize}
\shortv{
The sequent calculus we have just introduced could be extended with 
second-order quantifiers and additive logical connectives. For the sake 
of simplicity, however, we have kept the language of formulas very simple here. 
The interested reader can check \cite{phdlaurent} for a treatment of these 
connectives in a polarized setting or \cite{EV} for \remove{some }more details.}
\longv{
The sequent calculus we have just introduced could be extended with 
additive logical connectives. For the sake of simplicity, however, we have 
kept the language of formulas very simple here.
}

As already mentioned, $\BLLP$ proofs can be seen as obtained by decorating proofs from Laurent's
$\LLP$~\cite{laurent2003polarized} with resource polynomials. Given a proof $\proofone$, 
$\BLLPtoLLP{\proofone}$ is the $\LLP$ proof obtained by erasing all
resource polynomials occurring in $\proofone$. If $\proofone$ and $\prooftwo$ are two 
$\BLLP$ proofs, we write $\proofone\streq\prooftwo$ iff 
$\BLLPtoLLP{\proofone}\equiv\BLLPtoLLP{\prooftwo}$, i.e., iff $\proofone$ and $\prooftwo$ 
are two decorations of the same $\LLP$ proof.

Even if structural rules can be applied to all negative formulas, only certain
proofs will be copied or erased along the cut-elimination process, as we will soon realize.
A \emph{box} is any proof which ends with an occurrence of the $!$ rule. In non-polarized
systems, only boxes can be copied or erased, while here the process can be applied to 
\emph{$\TENSOR$-trees}, which are proofs inductively defined as follows:
\begin{varitemize}
\item
  Either the last rule in the proof is $\AX$ or $!$ or $1$;
\item
  or the proof is obtained from two $\TENSOR$-trees by applying the rule $\TENSOR$.
\end{varitemize}
A $\TENSOR$-tree is said to be \emph{closed} if it does not contain
any axiom nor any box having auxiliary doors (i.e., no formula in the context of the $!$ rules).

\subsection{Malleability}\label{sect:malleability}
The main reason for the strong (intensional) expressive power 
of $\BLL$~\cite{DalLago2009} is its \emph{malleability}: the conclusion of any proof $\proofone$ 
can be modified in many different ways without altering its structure.
Malleability is not only crucial to make the system expressive, but also to
prove that $\BLLP$ enjoys cut-elimination. In this section, we give four 
different ways of modifying a sequent in such a way as to preserve its 
derivability. Two of them are anyway expected and also hold in $\BLL$, while
the other two only make sense in a polarized setting.

First of all, taking smaller formulas (i.e., more general --- cf. \cite[\S 3.3, p. 21]{GSS92TCS})
preserves derivability:
\begin{lemma}[Subtyping]\label{lem:subtyping}
  If $\proofone\pof\seq \contextone, \ltone$
  and $\ltone\SUPTYPE\lttwo$, then there is
  $\prooftwo\pof\seq\contextone,\lttwo$ such that $\proofone\streq\prooftwo$.
\end{lemma}
\longv{
  \begin{proof}
    By a simple induction on $\proofone$. The crucial cases:
    \begin{varitemize}
    \item
      If the last rule used is an axiom:
      \begin{prooftree}
        \AxiomC{$\nltone \SUBTYPE \nlttwo$}
        \AxiomC{$\NOT{\nlttwo} \SUPTYPE \pltone$}
        \RightLabel{$\AX$}
        \BinaryInfC{$\seq \nltone, \pltone$}
      \end{prooftree}
      If $\lttwo\SUBTYPE\nltone$, then we
      know that $\lttwo\SUBTYPE\nltone\SUBTYPE\nlttwo$, from
      which it follows that $\lttwo\SUBTYPE\nlttwo$.
      We can thus take $\prooftwo$ as
      \begin{prooftree}
        \AxiomC{$\lttwo \SUBTYPE \nlttwo$}
        \AxiomC{$\NOT{\nlttwo} \SUPTYPE \pltone$}
        \RightLabel{$\AX$}
        \BinaryInfC{$\seq \lttwo, \pltone$}
      \end{prooftree}
      If $\lttwo\SUBTYPE\pltone$, then
      we know that $\lttwo\SUBTYPE\pltone\SUBTYPE\NOT{\nlttwo}$, from 
      which it follows that $\lttwo\SUBTYPE\NOT{\nlttwo}$.
      We can thus take $\prooftwo$ as
      \begin{prooftree}
        \AxiomC{$\lttwo \SUBTYPE \nlttwo$}
        \AxiomC{$\NOT{\nlttwo} \SUPTYPE \lttwo$}
        \RightLabel{$\AX$}
        \BinaryInfC{$\seq \lttwo, \nlttwo$}
      \end{prooftree}       
    \item 
      Suppose the last rule used is $!$:
      \begin{prooftree}
        \AxiomC{$\proofthree\pof\seq \negcontextone, \LABEL{\nfone}{\vartwo}{\polthree}$}
        \AxiomC{$\negcontexttwo\SUBTYPE\SUM_{\varone<\poltwo} \negcontextone$}
        \RightLabel{$!$}
        \BinaryInfC{$\seq \negcontexttwo, \LABEL{!_{\vartwo<\polthree}\nfone}{\varone}{\poltwo}$}
      \end{prooftree}
      If $\lttwo\SUBTYPE\LABEL{!_{\vartwo<\polthree}\nfone}{\varone}{\poltwo}$, then necessarily  
      $\lttwo=\LABEL{!_{\vartwo<\polfour}\nftwo}{\varone}{\polone}$, where $\nfone \SUPTYPE \nftwo$, 
      $\poltwo\resgeq \polone$ and $\polfour \resgeq \polthree$. Hence $\LABEL{\nfone}{\vartwo}{\polthree}
      \SUPTYPE\LABEL{\nftwo}{\vartwo}{\polfour}$ and, by induction hypothesis, there is $\prooffour$ such that 
      $\prooffour\streq\proofthree$ and $\prooffour\pof\seq\negcontextone, \LABEL{\nftwo}{\vartwo}{\polfour}$.
      As a consequence $\prooftwo$ can be simply defined as, because 
      $\SUM_{\varone<\poltwo}\negcontextone\SUBTYPE\SUM_{\varone<\polone}\negcontextone$:
      \begin{prooftree}
        \AxiomC{$\prooffour\pof\seq\negcontextone, \LABEL{\nftwo}{\vartwo}{\polfour}$}
        \AxiomC{$\negcontexttwo\SUBTYPE\SUM_{\varone<\polone} \negcontextone$}
        \RightLabel{$!$}
        \BinaryInfC{$\seq \negcontexttwo, \LABEL{!_{\vartwo<\polfour}\nftwo}{\varone}{\poltwo}$}
      \end{prooftree}
      If $\lttwo\SUBTYPE\nltone\in\negcontexttwo$, then we can just derive
      the thesis from transitivity of $\SUBTYPE$.
    \item 
      If the last rule used is $?d$:
      \begin{prooftree}
        \AxiomC{$\proofthree\pof\seq\negcontextone, \LABEL{\subst{\pfone}{\varone}{0}}{\vartwo}{\subst{\polthree}{\varone}{0}}$}
        \AxiomC{$\nltone\SUBTYPE\LABEL{?_{\vartwo<\polthree}\pfone}{\varone}{1}$}
        \RightLabel{$?d$}
        \BinaryInfC{$\seq \negcontextone,\nltone$}
      \end{prooftree}
      Then the induction hypothesis immediately yields the thesis.
    \end{varitemize}
    This concludes the proof.
  \end{proof}
}
Substituting resource variables for polynomials itself preserves typability:
\begin{lemma}[Substitution]\label{lem:subst}
  Let $\proofone\pof\seq\contextone$. Then there is a proof $\subst{\proofone}{\varone}{\polone}$
  of $\seq\subst{\contextone}{\varone}{\polone}$. Moreover, $\subst{\proofone}{\varone}{\polone}\streq \proofone$.
\end{lemma}
\longv{
\begin{proof}
By an easy induction on the structure of $\proofone$.
\end{proof}}
\shortv{
Both Lemma~\ref{lem:subtyping} and Lemma~\ref{lem:subst} can be proved by easy inductions
on the structure of $\proofone$.
}

\longv{
  \begin{lemma}\label{Lem:submon}
    $\LABEL{\typeone}{\varone}{\polone} \SUPTYPE \LABEL{\typetwo}{\varone}{\polone} \Rightarrow \LABEL{\subst{\typeone}{\varone}{\vartwo+\poltwo}}{\vartwo}{\polone} \SUPTYPE \LABEL{\subst{\typetwo}{\varone}{\vartwo+\poltwo}}{\vartwo}{\polone}.$
  \end{lemma}
  \begin{proof}
    $\LABEL{\typeone}{\varone}{\polone}\SUPTYPE \LABEL{\typetwo}{\varone}{\polone} \Rightarrow 
        \typeone \SUPTYPE \typetwo     \Rightarrow 
    \subst{\typeone}{\varone}{\vartwo+\poltwo}
    \SUPTYPE \subst{\typetwo}{\varone}{\vartwo+\poltwo} \Rightarrow 
\LABEL{\subst{\typeone}{\varone}{\vartwo+\poltwo}}{\vartwo}{\polone}
    \SUPTYPE \LABEL{\subst{\typetwo}{\varone}{\vartwo+\poltwo}}{\vartwo}{\polone}$
   \end{proof}
}  
As we have already mentioned, one of the key differences between ordinary Linear Logic and its polarized version is that in the latter,
arbitrary proofs can potentially be duplicated (and erased) along the cut-elimination process, while in the
former only special ones, namely boxes, can. This is, again, a consequence
of the fundamentally different nature of structural rules in the two systems. Since $\BLLP$ is a refinement
of $\LLP$, this means that the same phenomenon is expected. But beware: in a bounded setting, contraction is
not symmetric, i.e., the two copies of the proof $\proofone$ we are duplicating are not identical to $\proofone$.

What we need to prove, then, is that proofs can indeed be \emph{split} in $\BLLP$:
\longv{
  But preliminary to that is the following technical lemma:
  \begin{lemma}[Shifting Sums]\label{Lem:splitpro2}
    If $\sumlf{\varthree}{\poltwo}{\LABEL{\nftwo}{\vartwo}{\polthree}}=\LABEL{\nfone}{\vartwo}{\sum_{\varthree<\poltwo}\polthree}$,
    then the formula $\nltone=\subst{\LABEL{\nftwo}{\vartwo}{\polthree}}{\varthree}{\varthree+\poltwo}$ is such that
    $$
    \sumlf{\varthree}{\polone}{\nltone}=
    \LABEL{\subst{\nfone}{\varone}{\varone+\sum_{\varthree<\poltwo}\polthree}}{\vartwo}{\sum_{\varthree<\polone}
      \subst{\polthree}{\varthree}{\varthree+\poltwo}}
    $$ 
  \end{lemma}
  \begin{proof}
    The fact $\sumlf{\varthree}{\poltwo}{\LABEL{\nftwo}{\vartwo}{\polthree}}$ exists implies that there exist $\nfone,\varone,\varfour$ such
    that
    $$
    \nftwo=\subst{\nfone}{\varone}{\vartwo+\sum_{\varfour<\varthree}\subst{\polthree}{\varthree}{\varfour}}
    $$
    and $\vartwo,\varthree\notin\FV{\nfone}$ and $\vartwo\notin\FV{\polthree}$. As a consequence:
    \begin{align*}
      \subst{\LABEL{\nftwo}{\vartwo}{\polthree}}{\varthree}{\varthree+\poltwo}
      &=\LABEL{\subst{\subst{\nfone}{\varone}{\vartwo+\sum_{\varfour<\varthree}
            \subst{\polthree}{\varthree}{\varfour}}}{\varthree}{\varthree+\poltwo}}{\vartwo}{\subst{\polthree}{\varthree}{\varthree+\poltwo}}\\
      &=\LABEL{\subst{\nfone}{\varone}{\vartwo+\sum_{\varfour<\varthree+\poltwo}
            \subst{\polthree}{\varthree}{\varfour}}}{\vartwo}{\subst{\polthree}{\varthree}{\varthree+\poltwo}}\\
      &=\LABEL{\subst{\nfone}{\varone}{\vartwo+\sum_{\varfour<\poltwo}\subst{\polthree}{\varthree}{\varfour}+\sum_{\varfour<\varthree}
            \subst{\polthree}{\varthree}{\varfour+\poltwo}}}{\vartwo}{\subst{\polthree}{\varthree}{\varthree+\poltwo}}\\
      &=\LABEL{\subst{\subst{\nfone}{\varone}{\varone+\sum_{\varfour<\poltwo}\subst{\polthree}{\varthree}{\varfour}}}{\varone}{\vartwo+\sum_{\varfour<\varthree}
            \subst{\polthree}{\varthree}{\varfour+\poltwo}}}{\vartwo}{\subst{\polthree}{\varthree}{\varthree+\poltwo}}\\
      &=\LABEL{\subst{\subst{\nfone}{\varone}{\varone+\sum_{\varfour<\poltwo}\subst{\polthree}{\varthree}{\varfour}}}{\varone}{\vartwo+\sum_{\varfour<\varthree}
            \subst{\subst{\polthree}{\varthree}{\varthree+\poltwo}}{\varthree}{\varfour}}}{\vartwo}{\subst{\polthree}{\varthree}{\varthree+\poltwo}}
    \end{align*}
    Call the last formula $\nltone$. As a consequence, $\sumlf{\varthree}{\polone}{\nltone}$ exists and is
    equal to
    $$
    \LABEL{\subst{\nfone}{\varone}{\varone+\sum_{\varfour<\poltwo}\subst{\polthree}{\varthree}{\varfour}}}{\vartwo}{\sum_{\varthree<\polone}\subst{\polthree}{\varthree}{\varthree+\poltwo}}.
    $$
    This concludes the proof.
  \end{proof}
}
\begin{lemma}[Splitting]\label{lem:splitting}
  If $\proofone\pof\seq \negcontextone, \LABEL{\pfone}{\varone}{\polone}$ is a $\otimes$-tree
  and $\polone \resgeq \polthree+\polfour$ then there exist $\negcontexttwo,\negcontextthree$ such that
  $\prooftwo\pof\seq \negcontexttwo,\LABEL{\pfone}{\varone}{\polthree}$,
  $\proofthree\pof\seq \negcontextthree,\LABEL{\subst{\pfone}{\varone}{\vartwo+\polthree}}{\vartwo}{\polfour}$. Moreover,
  $\negcontextone\SUBTYPE\negcontexttwo\contsum\negcontextthree$ and $\prooftwo \streq \proofone\streq \proofthree$.
\end{lemma}
\longv{
\begin{proof}
  By induction on $\proofone$:
  \begin{varitemize}
  \item If the last rule used is an axiom then it is in the form
    \begin{prooftree}
      \AxiomC{$\LABEL{\nfone_1}{\varone}{\poltwo_1}\SUBTYPE \LABEL{\nftwo}{\varone}{\polfive}$}
      \AxiomC{$\LABEL{\NOT{\nftwo}}{\varone}{\polfive}\SUPTYPE \LABEL{\pfone}{\varone}{\polone}$}
      \RightLabel{$\AX$}
      \BinaryInfC{$\seq \LABEL{\nfone}{\varone}{\poltwo}, \LABEL{\pfone}{\varone}{\polone}$}
    \end{prooftree}
    for some $\nftwo, \polfive$. We know that
    $$
    \polthree+\polfour\resleq \polone\resleq \polfive \resleq \poltwo.
    $$
    Observe that, we can form the following derivations
    \begin{prooftree}
      \AxiomC{$\LABEL{\nfone_1}{\varone}{\polthree}\SUBTYPE \LABEL{\nftwo}{\varone}{\polthree}$}
      \AxiomC{$\LABEL{\NOT{\nftwo}}{\varone}{\polthree}\SUPTYPE \LABEL{\pfone}{\varone}{\polthree}$}
      \RightLabel{$\AX$}
      \BinaryInfC{$\seq \LABEL{\nfone_1}{\varone}{\polthree}, \LABEL{\pfone}{\varone}{\polthree}$}
    \end{prooftree}
    \begin{prooftree}
      \AxiomC{$\LABEL{\subst{\nfone}{\varone}{\vartwo+\polthree}}{\vartwo}{\polfour}\SUBTYPE 
        \LABEL{\subst{\nftwo}{\varone}{\vartwo+\polthree}}{\vartwo}{\polfour}$}
      \AxiomC{$\LABEL{\subst{\NOT{\nftwo}}{\varone}{\vartwo+\polthree}}{\vartwo}{\polfour}\SUPTYPE 
        \LABEL{\subst{\pfone}{\varone}{\vartwo+\polthree}}{\vartwo}{\polfour}$}
      \RightLabel{$\AX$}
      \BinaryInfC{$\seq \LABEL{\subst{\nfone}{\varone}{\vartwo+\polthree}}{\vartwo}{\polfour}, 
        \LABEL{\subst{\pfone}{\varone}{\vartwo+\polthree}}{\vartwo}{\polfour}$}
    \end{prooftree}
    where in building the second one we made use, in particular, of Lemma~\ref{Lem:submon}.
  \item 
    If the last rule used is $\TENSOR$ then we can write $\proofone$ as 
    \begin{prooftree}
      \AxiomC{$\prooffour_1\pof\seq\negcontextone_1, \LABEL{\pfone_1}{\varone}{\polone_1}$}
      \AxiomC{$\prooffour_2\pof\seq\negcontextone_2, \LABEL{\pfone_2}{\varone}{\polone_2}$}
      \RightLabel{$\TENSOR$}
      \BinaryInfC{$\seq \negcontextone_1, \negcontextone_2,\LABEL{\pfone_1 \TENSOR \pfone_2}{\varone}{\polone}$}
    \end{prooftree}
    where $\polone\resleq\polone_1$ and $\polone\resleq\polone_2$. As a consequence,
    $\polone_1\resgeq\poltwo+\polthree$ and $\polone_2\resgeq\poltwo+\polthree$, and we can thus apply the induction
    hypothesis to $\prooffour_1,\prooffour_2$ easily reaching the thesis.
  \item
    If the last rule used is promotion $!$ then $\proofone$ has the following shape:
    \begin{prooftree}
      \AxiomC{$\prooffour\pof\seq \negcontextone, \LABEL{\nfone}{\varthree}{\poltwo}$}
      \AxiomC{$\negcontexttwo\SUBTYPE\SUM_{\varone<\polthree}\negcontextone$}
      \RightLabel{$!$}
      \BinaryInfC{$\seq \negcontexttwo, \LABEL{!_{\varthree<\poltwo}\nfone}{\varone}{\polone}$}
    \end{prooftree}
    Then $\prooftwo$ is simply
    \begin{prooftree}
      \AxiomC{$\prooffour\pof\seq \negcontextone, \LABEL{\nfone}{\varthree}{\poltwo}$}
      \RightLabel{$!$}
      \UnaryInfC{$\seq\SUM_{\varone<\polthree}\negcontextone, \LABEL{!_{\varthree<\poltwo}\nfone}{\varone}{\polthree}$}
    \end{prooftree}
    About $\proofthree$, observe that $\subst{\prooffour}{\varone}{\vartwo+\polthree}$ has
    conclusion
    $$
    \seq \subst{\negcontextone}{\varone}{\vartwo+\polthree}, 
      \LABEL{\subst{\nfone}{\varone}{\vartwo+\polthree}}{\varthree}{\subst{\poltwo}{\varone}{\vartwo+\polthree}}
    $$
    By Lemma~\ref{Lem:splitpro2}, it is allowed to form $\sumlf{\vartwo}{\polfour}{\subst{\negcontextone}{\varone}{\vartwo+\polthree}}$.
    As a consequence, $\proofthree$ is
    \begin{prooftree}
      \AxiomC{$\seq \subst{\negcontextone}{\varone}{\vartwo+\polthree}, 
        \LABEL{\subst{\nfone}{\varone}{\vartwo+\polthree}}{\varthree}{\subst{\poltwo}{\varone}{\vartwo+\polthree}}$}
      \RightLabel{$!$}
      \UnaryInfC{$\seq \sumlf{\vartwo}{\polfour}{\subst{\negcontextone}{\varone}{\vartwo+\polthree}}, 
        \LABEL{\subst{(!_{\varthree<\poltwo}\nfone)}{\varone}{\vartwo+\polthree}}{\vartwo}{\polfour}$}
    \end{prooftree}
    Observe that the conclusions of $\prooftwo$ and $\proofthree$ are in the appropriate relation, again because of Lemma~\ref{Lem:splitpro2}.
  \end{varitemize}
This concludes the proof.
\end{proof}
}
Observe that not every proof can be split, but only $\TENSOR$-trees can.
\shortv{
The proof of Lemma~\ref{lem:splitting} is not
trivial and requires some auxiliary results (see~\cite{EV} for more details).
}
A parametric version of splitting
is also necessary here:
\begin{lemma}[Parametric Splitting]\label{Lem:parsplitting}
  If $\proofone\pof\seq \negcontextone, \LABEL{\pfone}{\varone}{\polone}$, 
  where $\proofone$ is a $\TENSOR$-tree and $\polone \resgeq \sum_{\varone<\polthree}\polfour$, then there exists,
  $\prooftwo\pof\seq \negcontexttwo, \LABEL{\pfone}{\varone}{\polfour}$
  where 
  $\sum_{\varone<\polthree}\negcontexttwo\SUPTYPE\negcontextone$.
  and $\prooftwo \streq \proofone$.
\end{lemma}
While splitting allows to cope with duplication, parametric splitting implies that an arbitrary $\TENSOR$-tree
proof can be modified so as to be lifted into a box through one of its auxiliary doors 
Please observe that $\polone^\proofone$ continues to be such an upper bound even if any natural number 
is substituted for any of its free variables, \remove{as }an easy consequence of Lemma~\ref{lem:subst}.
The following is useful when dealing with cuts involving the rule $?d$:
\begin{lemma}\label{lem:tensornew}
  If $\poltwo\resgeq 1$, then $\sumlf{\varthree}{\poltwo}{\LABEL{\nftwo}{\vartwo}{\polthree}}
  \SUBTYPE\subst{\LABEL{\nftwo}{\vartwo}{\polthree}}{\varthree}{0}$.
\end{lemma}
\longv{
  \begin{proof}
    By hypothesis, we have that $\sumlf{\varthree}{\poltwo}{\LABEL{\nftwo}{\vartwo}{\polthree}}=\LABEL{\nfone}{\vartwo}{\polone}$ for some
    $\nfone,\vartwo,\polone$. As a consequence
    $$
    \nftwo\equiv\subst{\nfone}{\varone}{\vartwo+\sum_{\varfour<\varthree}\subst{\polone}{\varthree}{\varfour}}.
    $$
    Now:
    \begin{align*}
      \subst{\LABEL{\nftwo}{\vartwo}{\polthree}}{\varthree}{0}&\equiv\LABEL{\subst{\nfone}{\varone}{\vartwo+\sum_{\varfour<0}\subst{\polone}{\varthree}{\varfour}}}{\vartwo}{\subst{\polthree}{\varthree}{0}}\\
      &\equiv\LABEL{\subst{\nfone}{\varone}{\vartwo}}{\vartwo}{\subst{\polthree}{\varthree}{0}}\equiv\LABEL{\nfone}{\varone}{\subst{\polthree}{\varthree}{0}}\SUPTYPE\LABEL{\nfone}{\varone}{\sum_{\varthree<\poltwo}\polthree}
    \end{align*}
    This concludes the proof.
  \end{proof}
}
\section{Cut Elimination}
In this Section, we \longv{show how a cut-elimination procedure for $\BLLP$ can be defined}\shortv{give some ideas about how cuts can be eliminated from $\BLLP$ proofs}. 
\longv{
We start by
showing how \emph{logical} cuts can be reduced, where a cut is logical when the two immediate subproofs
end with a rule introducing the formula involved in the cut. 
We describe how logical cuts can be reduced in the critical cases in Figure~\ref{fig:logcutelim}, which needs to be further explained:
\begin{figure*}
  \fbox{
  \begin{minipage}{.98\textwidth}
  \centering
    \vspace{5pt} \ \\
    \textsl{Multiplicatives}\\
    \vspace{5pt}
    {\scriptsize
      \AxiomC{$\seq \contextone, \LABEL{\nfone}{\varone}{\polone}, \LABEL{\nftwo}{\varone}{\poltwo}$}
      \RightLabel{$\PAR$}
      \UnaryInfC{$\seq \contextone, \LABEL{\nfone\PAR\nftwo}{\varone}{\polfive}$}
      \AxiomC{$\seq \negcontextone, \LABEL{\NOT{\nfone}}{\varone}{\polthree}$}
      \AxiomC{$\seq \negcontexttwo, \LABEL{\NOT{\nftwo}}{\varone}{\polfour}$}
      \RightLabel{$\TENSOR$}
      \BinaryInfC{$\seq \negcontextone, \negcontexttwo, \LABEL{\NOT{\nfone}\TENSOR\NOT{\nftwo}}{\varone}{\polfive}$}
      \RightLabel{$\CUT$}
      \BinaryInfC{$\seq \contextone, \negcontextone, \negcontexttwo$}
      \DisplayProof}
    $\logred$
    {\scriptsize
      \AxiomC{$\seq \contextone, \LABEL{\nftwo}{\varone}{\polfive}, \LABEL{\nfone}{\varone}{\polfive}$}
      \AxiomC{$\seq \negcontextone, \LABEL{\NOT{\nfone}}{\varone}{\polfive}$}
      \RightLabel{$\CUT$}
      \BinaryInfC{$\seq \contextone, \negcontextone, \LABEL{\nftwo}{\varone}{\polfive}$}
      \AxiomC{$\seq \negcontexttwo, \LABEL{\NOT{\nftwo}}{\varone}{\polfive}$}
      \RightLabel{$\CUT$}
      \BinaryInfC{$\seq \contextone, \negcontextone, \negcontexttwo$}
      \DisplayProof}\\
    \vspace{10pt}
    \textsl{Dereliction}\\
    \vspace{5pt}
    {\scriptsize
      \AxiomC{$\proofone\pof\seq \negcontextone, \LABEL{\nfone}{\vartwo}{\polone}$}
      \RightLabel{$!$}
      \UnaryInfC{$\seq\negcontexttwo, \LABEL{!_{\vartwo<\polone}\nfone}{\varone}{\poltwo}$}
      \AxiomC{$\prooftwo\pof\seq \negcontextthree, \LABEL{\subst{\NOT{\nftwo}}{\varone}{0}}{\vartwo}{\subst{\polthree}{\varone}{0}}$}
      \RightLabel{$?d$}
      \UnaryInfC{$\seq \negcontextthree, \LABEL{?_{\vartwo<\polone} \NOT{\nfone}}{\varone}{\poltwo}$}
      \RightLabel{$\CUT$}
      \BinaryInfC{$\seq \negcontexttwo, \negcontextthree$}
      \DisplayProof
      $\logred$
      \AxiomC{$\proofthree\pof\seq\negcontexttwo, \LABEL{\subst{\nfone}{\varone}{0}}{\vartwo}{\subst{\polone}{\varone}{0}}$}
      \AxiomC{$\prooffour\pof\seq\negcontextthree,\LABEL{\subst{\nfone}{\varone}{0}}{\vartwo}{\subst{\polone}{\varone}{0}}$}
      \RightLabel{$\CUT$}
      \BinaryInfC{$\seq\negcontexttwo, \negcontextthree$}
      \DisplayProof}\\
    \vspace{10pt}
    \textsl{Contraction}\\
    \vspace{5pt}
    {\scriptsize
    \AxiomC{$\proofone\pof\seq\negcontextone, \LABEL{\NOT{\nfone}}{\varone}{\polthree}$}
    \AxiomC{$\prooftwo\pof\seq \contextone, \LABEL{\nfthree}{\varone}{\polone}, \LABEL{\subst{\nfthree}{\varone}{\vartwo+\polone}}{\vartwo}{\poltwo}$}
    \RightLabel{$?c$}
    \UnaryInfC{$\seq \contextone, \LABEL{\nfone}{\varone}{\polthree}$}
    \RightLabel{$\CUT$}
    \BinaryInfC{$\seq \negcontextone, \contextone$}
    \DisplayProof\\
    \vspace{3pt}
    $\logred$\\
    \vspace{3pt}
    \AxiomC{$\prooffour\pof\seq \negcontextthree, \LABEL{\subst{\NOT{\nfthree}}{\varone}{\vartwo+\polone}}{\vartwo}{\poltwo}$}
    \AxiomC{$\proofthree\pof\seq \negcontexttwo, \LABEL{\NOT{\nfthree}}{\varone}{\polone}$}
    \AxiomC{$\proofone\pof\seq \contextone, \LABEL{\nfthree}{\varone}{\polone}, \LABEL{\subst{\nfthree}{\varone}{\varone+\polone}}{\vartwo}{\poltwo}$}
    \RightLabel{$\CUT$}
    \BinaryInfC{$\seq \negcontexttwo,\contextone, \LABEL{\subst{\nfthree}{\varone}{\varone+\polone}}{\vartwo}{\poltwo},\LABEL{\subst{\nfone}{\varone}{\varone+\polone}}{\vartwo}{\poltwo}$}
    \RightLabel{$\CUT$}
    \BinaryInfC{$\seq \negcontexttwo,\negcontextthree,\contextone$}
    \doubleLine
    \RightLabel{$?c$}
    \UnaryInfC{$\seq \negcontextone, \contextone$}
    \DisplayProof}\\
    \vspace{10pt}
    \textsl{Digging}\\
    \vspace{5pt}
    {\scriptsize
    \AxiomC{$\proofone\pof\seq\negcontextone,\LABEL{\NOT{\nfone}}{\varone}{\polthree}$}
    \AxiomC{$\prooftwo\pof\seq \negcontextthree,\LABEL{\nfthree}{\vartwo}{\polfour},\LABEL{\nftwo}{\vartwo}{\polone}$}
    \RightLabel{$!$}
    \UnaryInfC{$\seq\negcontexttwo,\LABEL{\nfone}{\varone}{\polthree},\LABEL{!_{\vartwo<\polone}\nftwo}{\varone}{\poltwo}$}
    \RightLabel{$\CUT$}
    \BinaryInfC{$\seq\negcontextone,\negcontexttwo,\LABEL{!_{\vartwo<\polone}\nftwo}{\varone}{\poltwo}$}    
    \DisplayProof
    \vspace{3pt}
    $\logred$
    \vspace{3pt}
    \AxiomC{$\proofthree \pof \seq \negcontextfour, \LABEL{\NOT{\nfthree}}{\vartwo}{\polfour}$}
    \AxiomC{$\prooftwo\pof\seq \negcontextthree,\LABEL{\nfthree}{\vartwo}{\polfour},\LABEL{\nftwo}{\vartwo}{\polone}$}
        \RightLabel{$\CUT$}
    \BinaryInfC{$\seq \negcontextfour, \negcontextthree, \LABEL{\nftwo}{\vartwo}{\polone}$}
    \RightLabel{$!$}
    \UnaryInfC{$\seq\negcontextone,\negcontexttwo,\LABEL{!_{\vartwo<\polone}\nftwo}{\varone}{\poltwo}$}    
    \DisplayProof}\\
  \vspace{5pt} \ \\
  \end{minipage}}
  \caption{Some Logical Cut-Elimination Steps}\label{fig:logcutelim}
\end{figure*}
\begin{varitemize}
\item
  When reducing  multiplicative logical cuts, we extensively use the Subtyping Lemma.
\item
  In the dereliction reduction step, $\subst{\proofone}{\varone}{0}$ (obtained through Lemma~\ref{lem:subst}) has conclusion
  $\seq\subst{\negcontextone}{\varone}{0}, \LABEL{\subst{\nfone}{\varone}{0}}{\vartwo}{\subst{\polone}{\varone}{0}}$. 
  By Lemma~\ref{lem:tensornew}, $\negcontexttwo\SUBTYPE\sum_{\varone<\poltwo}\negcontextone\SUBTYPE\subst{\negcontextone}{\varone}{0}$,
  and as a consequence, there is $\proofthree\pof\seq\negcontexttwo,\LABEL{\subst{\nfone}{\varone}{0}}{\vartwo}{\subst{\polone}{\varone}{0}}$. 
  From $\LABEL{?_{\vartwo<\polone}\NOT{\nfone}}{\varone}{\poltwo}\SUPTYPE\LABEL{?_{\vartwo<\polthree}\NOT{\nftwo}}{\varone}{1}$, it follows
  that $\LABEL{\subst{\NOT{\nftwo}}{\varone}{0}}{\vartwo}{\subst{\polthree}{\varone}{0}} \longv{\linebreak[1]} \SUPTYPE
  \LABEL{\subst{\nfone}{\varone}{0}}{\vartwo}{\subst{\polone}{\varone}{0}}$, and there is a proof
  $\prooffour\pof\seq \negcontextthree,\LABEL{\subst{\nfone}{\varone}{0}}{\vartwo}{\subst{\polone}{\varone}{0}}$.
\item
  In the contraction reduction step, we suppose that $\proofone$ is a $\otimes$-tree. Then we can apply Lemma \ref{lem:splitting} and Lemma \ref{lem:subtyping}, 
  and obtain $\proofthree\pof\seq\negcontexttwo, \LABEL{\NOT{\nfthree}}{\varone}{\polone}$ and 
  $\prooffour\pof\seq\negcontextthree, \LABEL{\subst{\NOT{\nfthree}}{\varone}{\vartwo+\polone}}{\vartwo}{\poltwo}$ 
  such that $\negcontexttwo\contsum\negcontextthree\SUPTYPE\negcontextone$.
\item
  In digging, by Lemma~\ref{Lem:parsplitting} from $\proofone$ we can find $\proofthree \pof \seq \negcontextfour, \LABEL{\NOT{\nfthree}}{\vartwo}{\polfour}$, 
  where $\negcontextone \SUBTYPE \SUM_{\varone < \poltwo} \negcontextfour$.
\end{varitemize}}
\shortv{\emph{Logical cuts} (i.e., those in which the two immediate subproofs
end with a rule introducing the formula involved in the cut) 
can be reduced as in $\LLP$~\cite{phdlaurent}, but exploiting malleability whenever polynomials need to be
modified. This defines the reduction relation $\logred$ (see~\cite{EV} for more details).}
All instances of the $\CUT$ rule which are not logical are said to be \emph{commutative}, and induce
a relation \shortv{$\commeq$} on proofs. 
\longv{
As an example, the proof
\begin{prooftree}
  \AxiomC{$\proofone\pof\seq\contextone,\nltone,\LABEL{\nfone}{\varone}{\polone}, \LABEL{\nftwo}{\varone}{\poltwo}$}
  \RightLabel{$\PAR$}
  \UnaryInfC{$\seq \contextone,\nltone,\LABEL{\nfone\PAR\nftwo}{\varone}{\polthree}$}
  \AxiomC{$\prooftwo\pof\seq\negcontextone,\NOT{\nltone}$}
  \RightLabel{$\CUT$}
  \BinaryInfC{$\seq\contextone,\negcontextone,\LABEL{\nfone\PAR\nftwo}{\varone}{\polthree}$}
\end{prooftree}
is equivalent to
\begin{prooftree}
  \AxiomC{$\proofone\pof\seq\contextone,\nltone,\LABEL{\nfone}{\varone}{\polone}, \LABEL{\nftwo}{\varone}{\poltwo}$}
  \AxiomC{$\prooftwo\pof\seq\negcontextone,\NOT{\nltone}$}
  \RightLabel{$\CUT$}
  \BinaryInfC{$\seq\contextone,\negcontextone,\LABEL{\nfone}{\varone}{\polone}, \LABEL{\nftwo}{\varone}{\poltwo}$}
  \RightLabel{$\PAR$}
  \UnaryInfC{$\seq \contextone,\negcontextone,\LABEL{\nfone\PAR\nftwo}{\varone}{\polthree}$}
\end{prooftree}
This way we can define an equivalence relation $\commeq$ on the space of proofs.}
In general, not all cuts in a proof are logical, but any cut can be turned into a logical one:
\begin{lemma}\label{lemma:logcut}
  Let $\proofone$ be any proof containing an occurrence of the rule $\CUT$. Then, there are two proofs $\prooftwo$ and $\proofthree$
  such that $\proofone\commeq\prooftwo\logred\proofthree$\remove{. Moreover, }, where $\prooftwo$ can be effectively obtained from $\proofone$.
\end{lemma}
The proof of Lemma~\ref{lemma:logcut} goes as follows: given any instance of the $\CUT$ rule
\begin{prooftree}
  \AxiomC{$\proofone\pof \seq \contextone, \LABEL{\nfone}{\varone}{\polone}$}
  \AxiomC{$\prooftwo\pof \seq \negcontextone, \LABEL{\pfone}{\varone}{\polone}$}
  \RightLabel{$\CUT$}
  \BinaryInfC{$\seq \contextone, \negcontextone$}
\end{prooftree}
consider the path (i.e., the sequence of formula occurrences) starting from $\LABEL{\nfone}{\varone}{\polone}$ and going
upward inside $\proofone$, and the path starting from $\LABEL{\pfone}{\varone}{\polone}$ and going upward
inside $\prooftwo$. Both paths end either at an $\AX$ rule or at an instance of a rule introducing the main
connective in $\nfone$ or $\pfone$. The game to play is then to show that these two paths can always be \emph{shortened} by
way of commutations, thus exposing the underlying logical cut.

Lemma~\ref{lemma:logcut} is implicitly defining a cut-elimination procedure: given any instance of the $\CUT$ rule, turn it into
a logical cut by the procedure from Lemma~\ref{lemma:logcut}, then fire it. This way we are implicitly defining another reduction
relation $\cutelimarr$. The next question is the following: is this procedure going to terminate for every proof $\proofone$ 
(i.e., is $\cutelimarr$ strongly, or weakly, normalizing)? How many steps does it take to turn $\proofone$ to its cut-free form?

Actually, $\cutelimarr$ produces reduction sequences of very long length, but is anyway strongly normalizing. A relatively
easy way to prove it goes as follows: any $\BLLP$ proof $\proofone$ corresponds to a $\LLP$ sequent calculus proof $\BLLPtoLLP{\proofone}$, 
and the latter itself corresponds to a polarized proof net $\BLLPtoLLPpn{\proofone}$~\cite{laurent2003polarized}. Moreover, $\proofone\cutelimarr\prooftwo$
implies that $\BLLPtoLLPpn{\proofone}\mapsto\BLLPtoLLPpn{\prooftwo}$, where $\mapsto$ is the canonical cut-elimination
relation on polarized proof-nets. Finally, $\BLLPtoLLPpn{\proofone}$ is identical to $\BLLPtoLLPpn{\prooftwo}$
whenever $\proofone\commeq\prooftwo$. Since $\mapsto$ is known to be strongly normalizing, $\cutelimarr$ does not admit infinite
reduction sequences:
\begin{proposition}[Cut-Elimination]
The relation $\cutelimarr$ is strongly normalizing.
\end{proposition}
This does not mean that cut-elimination can be performed
in (reasonably) bounded time. Already in $\BLL$ this can take hyperexponential time: the
whole of Elementary Linear Logic~\cite{Girard98IC} can be embedded into it.
\subsection{Soundness}\label{subsec:soundness}
To get a soundness result, then, we somehow need to restrict the underlying
reduction relation $\cutelimarr$. Following~\cite{GSS92TCS}, one could indeed define a subset of 
$\cutelimarr$ just by imposing that in dereliction, contraction, or box cut-elimination steps, 
the involved $\TENSOR$-trees are closed. Moreover, we could stipulate that reduction is external, 
i.e., it cannot take place inside boxes. Closed and external reduction, however, is not enough to
simulate head-reduction in the $\lambda\mu$-calculus, and not being able to reduce under the
scope of $\mu$-abstractions does not make much sense anyway. We are forced, then, to consider an
extension of closed reduction. The fact that this new notion of reduction still guarantees
polynomial bounds is technically a remarkable strengthening with respect to $\BLL$'s Soundness 
Theorem~\cite{GSS92TCS}.

There is a quite natural notion of \emph{downward} path in proofs: from any occurrence of a 
negative formula $\nltone$, just proceed downward until you either find (the main premise
of) a $\CUT$ rule, or a conclusion. In the first case, the occurrence of $\nltone$ is said 
to be \emph{active}, in the second it is said to be \emph{passive}. Proofs can then
be endowed with a new notion of reduction: all dereliction, contraction or box digging cuts
can be fired only if the negative formula occurrences in its rightmost argument are all passive. 
In the literature, this is sometimes called a \emph{special cut} (e.g.~\cite{Baillot11}). Moreover,
reduction needs to be external, as usual. This notion of reduction, as we will see, is enough to 
mimic head reduction, and is denoted  with $\redcutelimarr$.

The next step consists in \remove{attributing} associating a weight, in the form of a resource polynomial, to
every proof, similarly to what happens in $\BLL$. The \emph{pre-weight} $\pfd{\proofone}$ 
of a proof $\proofone$ with conclusion $\seq\ltone_1,\ldots,\ltone_n$ consists in:
\begin{varitemize}
\item
  a resource polynomial $\polone^\proofone$.
\item
  $n$ disjoints sets of resource variables $\srvone_1^\proofone,\ldots,\srvone_n^\proofone$, each corresponding to
  a formula in $\ltone_1,\ldots,\ltone_n$; if this does not cause ambiguity, the
  set of resource variables corresponding to a formula $\ltone$ will be denoted by $\srvp{\proofone}{\ltone}$.
  Similarly for $\srvp{\proofone}{\contextone}$, where $\contextone$ is a multiset of formulas.
\end{varitemize}
If $\proofone$ has pre-weight $\polone^\proofone,\srvone_1^\proofone,\ldots,\srvone_n^\proofone$, then the \emph{weight} $\poltwo^\proofone$ of 
$\proofone$ is simply $\polone^\proofone$ where, however, all the variables in $\srvone_1^\proofone,\ldots,\srvone_n^\proofone$ are
substituted with $0$: $\subst{\polone^\proofone}{\cup_{i=1}^n\srvone_i^\proofone}{0}$.
\longv{
The pre-weight of a proof $\proofone$ is defined by induction on the structure of $\proofone$, following
the rules in Figure~\ref{figure:preweights}.
\begin{figure*}
\centering
\begin{tabular}{|c|c|}\hline\hline
  \vspace{-10pt} \\ 
  $\proofone$ & $\pfd{\proofone}$\\ 
  \vspace{-20pt} \\ 
  \\ \hline\hline \vspace{-5pt} & \\
  \longv{\scriptsize}
  \AxiomC{$\LABEL{\nftwo}{\varone}{\poltwo} \SUBTYPE \LABEL{\nfone}{\varone}{\polone}$}
  \AxiomC{$\LABEL{\NOT{\nfone}}{\varone}{\polone} \SUPTYPE \LABEL{\pfone}{\varone}{\polthree}$}
  \RightLabel{$\AX$}
  \BinaryInfC{$\seq \LABEL{\nftwo}{\varone}{\poltwo},\LABEL{\pfone}{\varone}{\polthree}$}
  \DisplayProof
  &
  \longv{\scriptsize}
  $\{\vartwo\},\emptyset,\vartwo$
  \\ \vspace{-5pt} & \\
  \longv{\scriptsize}
  \AxiomC{$\prooftwo\pof\seq \contextone, \LABEL{\nfone}{\varone}{\polone}$}
  \AxiomC{$\proofthree\pof\seq \negcontextone, \LABEL{\NOT{\nfone}}{\varone}{\polone}$}
  \RightLabel{$\CUT$}
  \BinaryInfC{$\seq \contextone, \negcontextone$}
  \DisplayProof
  &
  \longv{\scriptsize}
  $\srvp{\prooftwo}{\contextone},\srvp{\proofthree}{\negcontextone},
    \subst{\polone^\prooftwo}{\srvp{\prooftwo}{\LABEL{\nfone}{\varone}{\polone}}}{1}+
    \subst{\polone^\proofthree}{\srvp{\proofthree}{\LABEL{\NOT{\nfone}}{\varone}{\polone}}}{1}$
  \\ \vspace{-5pt} & \\
  \longv{\scriptsize}
  \AxiomC{$\prooftwo\pof\seq \contextone, \LABEL{\nfone}{\varone}{\polone}, \LABEL{\nftwo}{\varone}{\poltwo}$}
  \AxiomC{$\polone\resleq\polthree$\quad$\poltwo\resleq\polthree$}
  \RightLabel{$\PAR$}
  \BinaryInfC{$\seq \contextone, \LABEL{\nfone\PAR\nftwo}{\varone}{\polthree}$}
  \DisplayProof
  &
  \longv{\scriptsize}
  $
  \srvp{\prooftwo}{\contextone},\srvp{\prooftwo}{\LABEL{\nfone}{\varone}{\polone}}\cup
    \srvp{\prooftwo}{\LABEL{\nftwo}{\varone}{\polone}}\cup\{\vartwo\},\polone^\prooftwo+\vartwo
  $
  \\ \vspace{-5pt} & \\
  \longv{\scriptsize}
  \AxiomC{$\prooftwo\pof\seq \negcontextone, \LABEL{\pfone}{\varone}{\polone}$}
  \AxiomC{$\proofthree\pof\seq \negcontexttwo, \LABEL{\pftwo}{\varone}{\polone}$\quad$\polthree\resleq\polone$\quad$\polthree\resleq\poltwo$}
  \RightLabel{$\TENSOR$}
  \BinaryInfC{$\seq \negcontextone, \negcontexttwo, \LABEL{\pfone\TENSOR \pftwo}{\varthree}{\polthree}$}
  \DisplayProof
  &
  \longv{\scriptsize}
  $
  \srvp{\prooftwo}{\negcontextone},\srvp{\proofthree}{\negcontexttwo},
  \srvp{\prooftwo}{\LABEL{\pfone}{\varone}{\polone}}\cup
  \srvp{\proofthree}{\LABEL{\pftwo}{\varone}{\polone}},
    \polone^\prooftwo+\polone^\proofthree
  $
  \\ \vspace{-5pt} & \\
  \longv{\scriptsize}
  \AxiomC{$\prooftwo\pof\seq\nltone_1,\ldots,\nltone_n,\LABEL{\nftwo}{\varone}{\polone}$}
  \AxiomC{$\nlttwo_i\SUBTYPE\SUM_{y<q}\nltone_i$}
  \RightLabel{!}
  \BinaryInfC{$\seq\nlttwo_1,\ldots,\nlttwo_n,
    \LABEL{!_{\varone<\polone}\nftwo}{\vartwo}{\poltwo}$}
  \DisplayProof
  &
  \longv{\scriptsize}
  $
  \srvp{\prooftwo}{\nltone_1}\cup\{\vartwo_1\},\ldots,\srvp{\prooftwo}{\nltone_n}\cup\{\vartwo_n\},\polone\cdot\polone^\prooftwo+
    \vartwo_1+\ldots+\vartwo_n
  $
  \\ \vspace{-5pt} & \\
  \longv{\scriptsize}
  \AxiomC{$\prooftwo\pof\seq \negcontextone, \LABEL{\subst{\pfone}{\vartwo}{0}}{\varone}{\subst{\polone}{y}{0}}$}
  \AxiomC{$\nltone\SUBTYPE\LABEL{?_{\varone<\polone}\pfone}{\vartwo}{1}$}
  \RightLabel{$?d$}
  \BinaryInfC{$\seq \negcontextone,\nltone$}
  \DisplayProof
  &
  \longv{\scriptsize}
  $
  \srvp{\prooftwo}{\negcontextone},\srvp{\prooftwo}{\LABEL{\subst{\pfone}{\vartwo}{0}}{\varone}{\subst{\polone}{y}{0}}}\cup\{\vartwo\},
    \polone^\prooftwo+\vartwo
  $
  \\ \vspace{-5pt} & \\
  \longv{\scriptsize}
  \AxiomC{$\prooftwo\pof\seq \contextone \vphantom{\LABEL{\nfone}{\varone}{\polone}}$}
  \RightLabel{$?w$}
  \UnaryInfC{$\seq \contextone,\nltone$}
  \DisplayProof
  &
  \longv{\scriptsize}
  $
  \srvp{\prooftwo}{\contextone},\{\vartwo\}
  $
  \\ \vspace{-5pt} & \\
  \AxiomC{$\prooftwo\pof\seq \contextone,\nltone,\nlttwo$}
  \AxiomC{$\nltthree\SUBTYPE\nltone\contsum\nlttwo$}
  \RightLabel{$?c$}
  \BinaryInfC{$\seq \contextone,\nltthree$}
  \DisplayProof
  &
  \longv{\scriptsize}
  $
  \srvp{\prooftwo}{\contextone},\srvp{\prooftwo}{\nltone}\cup\srvp{\prooftwo}{\nlttwo}\cup\{\vartwo\}
  $
  \\ \vspace{-5pt} & \\
  \longv{\scriptsize}
  \AxiomC{$\prooftwo\pof\seq \contextone$}
  \RightLabel{$\bot$}
  \UnaryInfC{$\seq \contextone,\LABEL{\bot}{\varone}{\polone}$}
  \DisplayProof
  &
  \longv{\scriptsize}
  $
  \srvp{\prooftwo}{\contextone},\{\vartwo\},\polone^\prooftwo+\vartwo
  $
  \\ \vspace{-5pt} & \\
  \longv{\scriptsize}
  \AxiomC{$\vphantom{\seq \contextone}$}
  \RightLabel{$1$}
  \UnaryInfC{$\seq \LABEL{1}{\varone}{\polone}$}
  \DisplayProof
  &
  \longv{\scriptsize}
  $
  \emptyset,0
  $
  \\ \vspace{-5pt} & \\ \hline\hline
\end{tabular}
\vspace{5pt}
\caption{Pre-weights for Proofs.}\label{figure:preweights}
\end{figure*}
Please notice how any negative formula $\nltone$ in the conclusion of $\proofone$ is associated with some fresh variables, each
accounting for the application of a rule to it. When $\nltone$ is then applied to a cut, all these
variables are set to $1$. 
}
\shortv
{
The pre-weight of a proof $\proofone$ is defined by induction on the structure of $\proofone$ (see~\cite{EV} for more details). The idea is that
every occurrence of negative formulas is attributed a fresh variable, which later is instantiated with either $0$ (if the
formula is passive) or $1$ (if it is active).
}
This allows to discriminate between the case in which rules can ``produce'' time
complexity along the cut-elimination, and the case in which they do not. Ultimately, this leads to:
\begin{lemma}\label{lemma:monotone}
If $\proofone\commeq\prooftwo$, then $\poltwo^\proofone=\poltwo^\prooftwo$.
If $\proofone\redcutelimarr\prooftwo$, then $\poltwo^\prooftwo\reslt\poltwo^\proofone$. 
\end{lemma}
The main idea behind Lemma~\ref{lemma:monotone} is that even if the logical cut we perform when
going from $\proofone$ to $\prooftwo$ is ``dangerous'' (e.g. a contraction) \emph{and} the involved
$\otimes$-tree is not closed, the residual negative rules have null weight, because they are passive.

We can conclude that:
\begin{theorem}[Polystep Soundness]\label{theo:polystepbllp}
For every proof $\proofone$, if $\proofone\redcutelimarr^n\prooftwo$, then $n\leq\poltwo^\proofone$.
\end{theorem}
In a sense, then, the weight of any proof $\proofone$ is a resource polynomial which can be easily computed
from $\proofone$ \longv{(rules in \longv{Figure~\ref{figure:preweights}}\shortv{\cite{EV}} are anyway inductively defined)} but which is also
an upper bound on the number of logical cut-elimination steps separating $\proofone$ from its normal form.
Please observe that $\poltwo^\proofone$ continues to be such an upper bound even if any natural number 
is substituted for any of its free variables, \remove{as }an easy consequence of Lemma~\ref{lem:subst}.

Why then, are we talking about \emph{polynomial} bounds? In $\BLL$, and as a consequence also in 
$\BLLP$, one can write programs in such a way that the size of the input is reflected by a resource variable
occurring in its type. 
\shortv{
As an example, the type of (Church encodings of) binary strings of length at most $\varone$ 
could be the following in $\BLLP$:
$$
(\COATOMONE\arrow{}{1}\COATOMONE)\arrow{}{\varone}(\COATOMONE\arrow{}{1}\COATOMONE)\arrow{}{\varone}(\COATOMONE\arrow{}{1}\COATOMONE)
$$
(where $\nfone\arrow{}{\polone}\nftwo$ stands for $?_\polone\NOT{\nfone}\PAR\nftwo$). The weight, then, turns out to be a tool 
to study the behavior of terms seen as functions taking arguments of varying length. A more in-depth discussion about 
these issues is outside the scope of this paper.} 
Please refer to \cite{GSS92TCS}.
\section{A Type System for the $\lambda\mu$-Calculus}
We describe here a version of the $\lambda\mu$-calculus as introduced by de Groote~\cite{DBLP:conf/caap/Groote94}. Terms are as follows
$$
\termone,\termtwo\;::=\ \lvarone\midd\lambda\lvarone.\termone\midd\mu\mvarone.\termone\midd\ [\mvarone]\termone\midd(\termone)\termone,
$$
where $\lvarone$ and $\mvarone$ range over two infinite disjoint sets of variables (called $\lambda$-variables and $\mu$-variables, respectively).
In contrast with the $\lambda\mu$-calculus as originally formulated by Parigot~\cite{Parigot}, $\mu$-abstraction is not restricted to 
terms of the form $[\mvarone]\termone$ here.
\subsection{Notions of Reduction}
The reduction rules we consider are the following 
ones:
\newcommand{\tow}{\to_{\mathsf{w}}}
\newcommand{\toh}{\to_{\mathsf{h}}}
\newcommand{\lsubst}[3]{{#1}[{}^{#3}/{}_{#2}]}
\begin{align*}
  (\lambda x.\termone)\termtwo&\to_\beta \termone[{}^\termtwo/{}_x]; & & &
  (\mu \alpha.\termone)\termtwo&\to_\mu \mu\alpha.\termone[{}^{[\alpha](\termthree)\termtwo}/{}_{[\alpha]\termthree}]; & & &
  \mu\alpha.[\alpha]\termone&\to_\theta \termone; 
\end{align*}
where, as usual, $\to_\theta$ can be fired only if $\muvarone \not \in \FV{\termone}$. 
In the following, $\to$ is just $\to_{\beta\mu\theta}$. In so-called \emph{weak reduction}, 
denoted $\tow$, reduction simply cannot take place in the scope of binders,
while \emph{head reduction}, denoted $\toh$, is a generalization of the same concept from
pure $\lambda$-calculus~\cite{de1998environment}. Details are in Figure~\ref{fig:whred}.
\begin{figure*}
\fbox{
\begin{minipage}{.98\textwidth}
\vspace{5pt}
\begin{center}
  \AxiomC{$\termone\to\termtwo$}
  \UnaryInfC{$\termone\tow\termtwo$}
  \DisplayProof
  \hspace{1pt}
  \AxiomC{$\termone\tow\termtwo$}
  \UnaryInfC{$\termone\termthree\tow\termtwo\termthree$}
  \DisplayProof
  \hspace{1pt}
  \AxiomC{$\termone\tow\termtwo$}
  \UnaryInfC{$[\alpha]\termone\tow[\alpha]\termtwo$}
  \DisplayProof
  \longv{\\ \ \vspace{5pt} \\}\shortv{\hspace{1pt}}
  \AxiomC{$\termone\tow\termtwo$}
  \UnaryInfC{$\termone\toh\termtwo$}
  \DisplayProof
  \hspace{1pt}
  \AxiomC{$\termone\toh\termtwo$}
  \UnaryInfC{$\lambda\varone.\termone\toh\lambda\varone.\termtwo$}
  \DisplayProof
  \hspace{1pt}
  \AxiomC{$\termone\toh\termtwo$}
  \UnaryInfC{$\mu\alpha.\termone\toh\mu\alpha.\termtwo$}
  \DisplayProof
\end{center}
\vspace{-3pt}
\end{minipage}}
\caption{Weak and Head Notions of Reduction}\label{fig:whred}
\end{figure*}
Please notice how in head reduction, redexes can indeed be fired even if they lie in the scope of $\lambda$-or-$\mu$-abstractions, which,
however, cannot themselves be involved in a redex. This harmless restriction, which corresponds to taking the \emph{outermost} reduction
order, is needed for technical reasons that will become apparent soon.
\subsection{The Type System}
Following Laurent~\cite{laurent2003polarized}, types are just negative formulas. Not all of them can be used as types, however: in particular,
$\nfone\PAR\nftwo$ is a legal type only if $\nfone$ is in the form $?_{\varone<\polone}\NOT{\nfthree}$, and we use the following
abbreviation in this case: $\nfone \arrow{\varone}{\polone} \nftwo\equiv (?_{\varone<\polone}\NOT{\nfone})\PAR \nftwo$.
In particular, if $\nftwo$ is $\bot$ then
$\nfone \arrow{\varone}{\polone} \bot$ can be abbreviated as $\pneg{\varone}{\polone}{\nfone}$. 
\emph{Typing formulas} are negative formulas which are either $\bot$, or $\COATOMONE$, 
or in the form $\nfone \arrow{\varone}{\polone}\nftwo$ (where $\nfone$ and $\nftwo$
are typing formulas themselves). 
A \emph{modal formula} is one in the form $?_{\varone<\polone}\NOT{\nfone}$ (where $\nfone$ is a typing formula).
Please observe that all the constructions from Section~\ref{sect:polform} (including
labellings, sums, etc.) easily apply to typing formulas. Finally, we use the following
abbreviation for labeled modal formulas:
$\DLABEL{\vartwo}{\poltwo}{\nfone}{\varone}{\polone} \equiv \LABEL{?_{\vartwo<\poltwo}\NOT{\nfone}}{\varone}{\polone}$.

A \emph{typing judgement} is a statement in the form 
$\tj{\contextone}{\termone}{\nltone}{\contexttwo}$, where:
\begin{varitemize} 
\item
  $\contextone$ is a context assigning labelled modal formulas to
  $\lambda$-variables;
\item
  $\termone$ is a $\lambda\mu$-term;
\item
  $\nltone$ is a typing formula;
\item
  $\contexttwo$ is a context assigning labelled typing formulas to 
  $\mu$-variables.
\end{varitemize}
The way typing judgments are defined allows to see them as $\BLLP$ sequents. This way, again, various concepts
from Section~\ref{sect:sequents} can be lifted up from sequents to judgments, and this remarkably includes the subtyping relation
$\SUBTYPE$. 

Typing rules are in Figure~\ref{fig:typingadditive}.
\begin{figure*}
  \fbox{\shortv{\scriptsize}
  \begin{minipage}{.98\textwidth}
  \centering
    \ \vspace{10pt}
    \\
    \AxiomC{$1\resleq \polone, \subst{\polthree}{\vartwo}{0}\resleq \poltwo, \nftwo \SUBTYPE \subst{\nfone}{\vartwo}{0}$}
    \RightLabel{\textsf{var}}
    \UnaryInfC{$\contextone, \varone: \DLABEL{\varthree}{\polthree}{\nfone}{\vartwo}{\polone} \seq \varone: \LABEL{\nftwo}{\varthree}{\poltwo} \mid\contexttwo$}
    \DisplayProof
    \qquad
    \AxiomC{$\contextone, \varone : \DLABEL{\varthree}{\polfour}{\nfone}{\vartwo}{\polone} \seq \termone: \LABEL{\nftwo}{\vartwo}{\poltwo}\mid\contexttwo$}
    \AxiomC{$\polthree\resgeq \poltwo, \polthree\resgeq \polone$}
    \RightLabel{\textsf{abs}}
    \BinaryInfC{$\contextone \seq \lambda \varone.\termone: \LABEL{\nfone \arrow{\varthree}{\polfour}\nftwo}{\vartwo}{\polthree}\mid\contexttwo$}
    \DisplayProof
    \\
    \ \vspace{10pt}
    \\
    \AxiomC{$\contextthree \seq \termone : \LABEL{\nfone \arrow{\varone}{\polone} \nftwo}{\vartwo}{\poltwo}\mid\contextfive$}
    \AxiomC{$\contextfour \seq \termtwo : \LABEL{\nfone}{\varone}{\polone}\mid\contextsix$}
    \AxiomC{\parbox{120pt}{
        $
        \polnine \resgeq \poltwo\qquad \polten \resgeq \poltwo\\
        \contextone \SUBTYPE \contextthree \contsum \contextseven\qquad\contextseven \SUBTYPE \SUM_{\varnine<\polnine} \contextfour\\
        \contexttwo \SUBTYPE \contextfive \contsum \contexteight\qquad\contexteight \SUBTYPE \SUM_{\varnine<\polnine}\contextsix
        $}}
    \RightLabel{\textsf{app}}
    \TrinaryInfC{$\contextone \seq (\termone)\termtwo: \LABEL{\nftwo}{\vartwo}{\polten}\mid\contexttwo$}
    \DisplayProof
    \\
    \ \vspace{10pt}
    \\
    \AxiomC{$\contextone \seq \termone : \nltone\mid\muvarone: \nlttwo, \contexttwo$}
    \AxiomC{$\nltthree\SUBTYPE\nltone\contsum\nlttwo$}
    \RightLabel{$\mu$\textsf{-name}}
    \BinaryInfC{$\contextone \seq [\muvarone]\termone: \LABEL{\bot}{\varthree}{\poltwo}\mid\muvarone: \nltthree, \contexttwo$}
    \DisplayProof
    \qquad
    \AxiomC{$\contextone \seq \termone: \LABEL{\bot}{\varthree}{\poltwo}\mid\muvartwo: \nltone, \contexttwo$}
    \RightLabel{$\mu$\textsf{-abs}}
    \UnaryInfC{$\contextone \seq \mu \muvartwo \termone: \nltone\mid\contexttwo$}
    \DisplayProof
    \\
    \ \vspace{10pt}
  \end{minipage}}
\longv{\caption{(Additive) Type Assignment Rules}}\shortv{\caption{Type Assignment Rules}}\label{fig:typingadditive}
\end{figure*}
The typing rule for applications, in particular, can be seen as overly complicated. In fact, all premises except the first two are
there to allow the necessary degree of malleability for contexts, without which even subject reduction would be in danger.
Alternatively, one could consider an explicit subtyping rule, the price being the loss of syntax directness. Indeed, all malleability results
from Section~\ref{sect:malleability} can be transferred to the just defined type assignment system.
\subsection{Subject Reduction and Polystep Soundness}
The aim of this Section is to show that \emph{head} reduction preserves types, and as a corollary, that the number of reduction
steps to normal form is bounded by a polynomial, along the same lines as in Theorem~\ref{theo:polystepbllp}. Actually, the latter will
easily follow from the former, because so-called Subject Reduction will be formulated (and in a sense proved) with a precise correspondence
between type derivations and proofs in mind.

In order to facilitate this task, Subject Reduction is proved on a modified type-assignment system, called $\BLLPLMM$ which 
can be proved equivalent to $\BLLPLM$. The only fundamental difference between the two systems lies in how structural rules, 
i.e., contraction and weakening, are reflected into the type system. As we have already noticed, $\BLLPLM$ has an 
\emph{additive} flavour, since structural rules are implicitly applied in binary and $0$-ary typing rules. This, in 
particular, makes the system syntax directed and type derivations more compact. The only problem with this approach is 
that the correspondence between type derivations and proofs is too weak to be directly lifted to a dynamic level 
(e.g., one step in $\toh$ could correspond to possibly many steps in $\redcutelimarr$). In $\BLLPLMM$, on the contrary, 
structural rules are explicit, and turns it into a useful technical tool to prove properties of $\BLLPLM$. \shortv{The
rules of $\BLLPLMM$ are in~\cite{EV}.}

\longv{
$\BLLPLMM$'s typing judgments are precisely the ones of $\BLLPLM$. What changes are typing rules, which are in Figure~\ref{fig:typingmultiplicative}.
\begin{figure*}
\centering
    \fbox{
      \begin{minipage}{.98\textwidth}
        \begin{center}
          \ \vspace{5pt} \\
          \AxiomC{$1\resleq \polone, \subst{\polthree}{\vartwo}{0}\resleq \poltwo, \nftwo \SUBTYPE \subst{\nfone}{\vartwo}{0}$}
          \RightLabel{\textsf{var}}
          \UnaryInfC{$\varone: \DLABEL{\varthree}{\polthree}{\nfone}{\vartwo}{\polone} \seq \varone: \LABEL{\nftwo}{\varthree}{\poltwo} |$}
          \DisplayProof
          \qquad
          \AxiomC{$\contextone, \varone : \DLABEL{\varthree}{\polfour}{\nfone}{\vartwo}{\polone} \seq \termone: \LABEL{\nftwo}{\vartwo}{\poltwo} | \contexttwo$}
          \AxiomC{$\polthree\resgeq \poltwo, \polthree\resgeq \polone$}
          \RightLabel{\textsf{abs}}
          \BinaryInfC{$\contextone \seq \lambda \varone.\termone: \LABEL{\nfone \arrow{\varthree}{\polfour}\nftwo}{\vartwo}{\polthree} | \contexttwo$}
          \DisplayProof
          \\
          \ \vspace{10pt}
          \\
          \AxiomC{$\contextone \seq \termone : \LABEL{\nfone \arrow{\varone}{\polone} \nftwo}{\vartwo}{\poltwo} | \contexttwo$}
          \AxiomC{$\contextthree \seq \termtwo : \LABEL{\nfone}{\varone}{\polone} | \contextfour$}
          \AxiomC{\parbox{2.6cm}{
              $
              \polnine \resgeq \poltwo, \polten \resgeq \poltwo,\\
              \contextfive \SUBTYPE \SUM_{\varnine<\polnine} \contextthree,\\
              \contextsix \SUBTYPE \SUM_{\varnine<\polnine}\contextfour
              $}}
          \RightLabel{\textsf{app}}
          \TrinaryInfC{$\contextone, \contextfive \seq (\termone)\termtwo: \LABEL{\nftwo}{\vartwo}{\polten} | \contexttwo, \contextsix$}
          \DisplayProof
          \\
          \ \vspace{10pt}
          \\
          \AxiomC{$\contextone \seq \termone : \nltone | \contexttwo$}
          \RightLabel{$\mu$\textsf{-name}}
          \UnaryInfC{$\contextone \seq [\muvarone]\termone: \LABEL{\bot}{\varthree}{\poltwo} | \muvarone: \nltone, \contexttwo$}
          \DisplayProof
          \qquad
          \AxiomC{$\contextone \seq \termone: \LABEL{\bot}{\varthree}{\poltwo} | \muvartwo: \nltone, \contexttwo$}
          \RightLabel{$\mu$\textsf{-abs}}
          \UnaryInfC{$\contextone \seq \mu \muvartwo \termone: \nltone | \contexttwo$} 
          \DisplayProof
          \\
          \ \vspace{10pt}
          \\
          \AxiomC{$\contextone \seq \termone:\nltone \mid \contexttwo$}
          \RightLabel{$?w^\lambda$}
          \UnaryInfC{$\contextone, \vartwo:\nlttwo \seq \termone:\nltone \mid \contexttwo$}
          \DisplayProof
          \qquad
          \AxiomC{$\contextone, \varone: \nltone, \vartwo: \nlttwo \seq \termone: \nltfour \mid \contexttwo$}
          \AxiomC{$\nltthree \SUBTYPE \nltone \contsum \nlttwo$}
          \RightLabel{$?c^\lambda$}
          \BinaryInfC{$\contextone, \varthree: \nltthree \seq \subst{\subst{\termone}{\varone}{\varthree}}{\vartwo}{\varthree}:\nltfour \mid \contexttwo$}
          \DisplayProof
          \\
          \ \vspace{10pt}
          \\
          \AxiomC{$\contextone \seq \termone:\nltone \mid \contexttwo$}
          \RightLabel{$?w^\mu$}
          \UnaryInfC{$\contextone \seq \termone:\nltone \mid \contexttwo, \muvarone:\nlttwo $}
          \DisplayProof
          \qquad
          \AxiomC{$\contextone \seq \termone: \nltfour \mid \contexttwo, \muvarone: \nltone, \muvartwo: \nlttwo$}
          \AxiomC{$\nltthree \SUBTYPE \nltone \contsum \nlttwo$}
          \RightLabel{$?c^\mu$}
          \BinaryInfC{$\contextone\seq \subst{\subst{\termone}{\muvarone}{\muvarthree}}{\muvartwo}{\muvarthree}: \nltfour \mid \contexttwo, \muvarthree: \nltthree $}
          \DisplayProof\\
          \ \vspace{5pt} \\ 
        \end{center}
    \end{minipage}}
\caption{(Multiplicative) Type Assignment Rules}\label{fig:typingmultiplicative}
\end{figure*}}
Whenever derivability in one of the system needs to be distinguished from derivability on the
other, we will put the system's name in subscript position (e.g. $\tjp{\contextone}{\BLLPLMM}{\termone}{\nltone}{\contexttwo}$).
Not so surprisingly, the two $\BLLPLM$ and $\BLLPLMM$ type exactly the same class of terms:
\begin{lemma}
$\tjp{\contextone}{\BLLPLMM}{\termone}{\nltone}{\contexttwo}$ iff $\tjp{\contextone}{\BLLPLM}{\termone}{\nltone}{\contexttwo}$
\end{lemma}
\begin{proof}
The left-to-right implication follows from weakening and contraction lemmas for $\BLLPLM$, which are easy to prove. The
right-to-left implication is more direct, since additive $\textsf{var}$ and $\textsf{app}$ are multiplicatively derivable.\shortv{\qed}
\end{proof}

Given a $\BLLPLMM$ type derivation $\proofone$, one can define a $\BLLP$ proof $\pfd{\proofone}$ \longv{following the rules in Figure~\ref{fig:mapping},
which work by induction on the structure of $\proofone$.
\begin{figure*}
 \centering
    {\footnotesize
    \begin{tabular}{|c|c|}\hline\hline
      \vspace{-5pt} & \\ 
      $\proofone$ & $\pfd{\proofone}$\\ 
      \vspace{-5pt} & \\ 
      \hline\hline
      \vspace{-1pt} & \\ 
      \AxiomC{$1\resleq \polone, \subst{\polthree}{\vartwo}{0}\resleq \poltwo, \nftwo \SUBTYPE \subst{\nfone}{\vartwo}{0}$}
      \RightLabel{\textsf{var}}
      \UnaryInfC{$\varone: \DLABEL{\varthree}{\polthree}{\nfone}{\vartwo}{\polone} \seq \varone: \LABEL{\nftwo}{\varthree}{\poltwo} |$}
      \DisplayProof
      &
      \AxiomC{}
      \UnaryInfC{$\seq\LABEL{\subst{\NOT{\nfone}}{\vartwo}{0}}{\varthree}{\subst{\polthree}{\vartwo}{0}},\LABEL{\nftwo}{\varthree}{\poltwo}$}
      \UnaryInfC{$\seq\LABEL{?_{\varthree<\polthree}\NOT{\nfone}}{\vartwo}{\polone},\LABEL{\nftwo}{\varthree}{\poltwo}$}
      \DisplayProof
      \\ \vspace{-1pt} & \\
      \AxiomC{$\prooftwo\pof\contextone, \varone : \DLABEL{\varthree}{\polfour}{\nfone}{\vartwo}{\polone} \seq \termone: \LABEL{\nftwo}{\vartwo}{\poltwo} | \contexttwo$}
      \RightLabel{\textsf{abs}}
      \UnaryInfC{$\contextone \seq \lambda \varone.\termone: \LABEL{\nfone \arrow{\varthree}{\polfour}\nftwo}{\vartwo}{\polthree} | \contexttwo$}
      \DisplayProof
      &
      \AxiomC{$\pfd{\prooftwo}\pof\seq\contextone, \LABEL{?_{\varthree<\polfour}\nfone}{\vartwo}{\polone},\LABEL{\nftwo}{\vartwo}{\poltwo},\contexttwo$}
      \UnaryInfC{$\seq\contextone, \LABEL{?_{\varthree<\polfour}\nfone\parr\nftwo}{\vartwo}{\polthree},\contexttwo$}
      \DisplayProof
      \\ \vspace{-1pt} & \\
      \AxiomC{
        \parbox{120pt}{
          \begin{center}
          $\prooftwo\pof\contextone \seq \termone : \LABEL{\nfone \arrow{\varone}{\polone} \nftwo}{\vartwo}{\poltwo} | \contexttwo$\\ 
          $\proofthree\pof\contextthree \seq \termtwo : \LABEL{\nfone}{\varone}{\polone} | \contextfour$\end{center}}}
                \RightLabel{\textsf{app}}
      \UnaryInfC{$\contextone, \contextfive \seq (\termone)\termtwo: \LABEL{\nftwo}{\vartwo}{\polten} | \contexttwo, \contextsix$}
      \DisplayProof
      &
      \AxiomC{$\pfd{\prooftwo}\pof\seq\contextone,\LABEL{?_{\varone<\polone}\NOT{\nfone}\parr\nftwo}{\vartwo}{\poltwo},\contexttwo$}
      \AxiomC{$\pfd{\proofthree}\pof\seq\contextthree,\LABEL{\nfone}{\varone}{\polone},\contextfour$}
      \UnaryInfC{$\seq\contextfive,\LABEL{!_{\varone<\polone}\nfone}{\vartwo}{\polnine},\contextsix$}
      \AxiomC{$\seq\LABEL{\NOT{\nftwo}}{\vartwo}{\polten},\LABEL{\nftwo}{\vartwo}{\polten}$}
      \BinaryInfC{$\seq\contextfive,\LABEL{!_{\varone<\polone}\nfone\otimes\NOT{\nftwo}}{\vartwo}{\poltwo},\contextsix,\LABEL{\nftwo}{\vartwo}{\polten}$}
      \BinaryInfC{$\seq\contextone,\contextfive,\LABEL{\nftwo}{\vartwo}{\polten},\contexttwo,\contextsix$}
      \DisplayProof
      \\ \vspace{-1pt} & \\
      \AxiomC{$\prooftwo\pof\contextone \seq \termone : \nltone | \contexttwo$}
      \RightLabel{$\mu$\textsf{-name}}
      \UnaryInfC{$\contextone \seq [\muvarone]\termone: \LABEL{\bot}{\varthree}{\poltwo} | \muvarone: \nltone, \contexttwo$}
      \DisplayProof
      &
      \AxiomC{$\pfd{\prooftwo}\pof\contextone,\nltone,\contexttwo$}
      \UnaryInfC{$\pfd{\prooftwo}\pof\contextone,\LABEL{\bot}{\varthree}{\poltwo},\nltone,\contexttwo$}
      \DisplayProof
      \\ \vspace{-1pt} & \\
      \AxiomC{$\prooftwo\pof\contextone \seq \termone: \LABEL{\bot}{\varthree}{\poltwo} | \muvartwo: \nltone, \contexttwo$}
      \RightLabel{$\mu$\textsf{-abs}}
      \UnaryInfC{$\contextone \seq \mu \muvartwo \termone: \nltone | \contexttwo$} 
      \DisplayProof
      &
      \AxiomC{$\pfd{\prooftwo}\pof\seq\contextone,\LABEL{\bot}{\varthree}{\poltwo},\nltone,\contexttwo$}
      \AxiomC{$\seq\LABEL{1}{\varthree}{\poltwo}$}
      \BinaryInfC{$\seq\contextone,\nltone,\contexttwo$}
      \DisplayProof
      \\ \vspace{-1pt} & \\
      \AxiomC{$\prooftwo\pof\contextone \seq \termone:\nltone \mid \contexttwo$}
      \RightLabel{$?w^\lambda$}
      \UnaryInfC{$\contextone, \vartwo:\nlttwo \seq \termone:\nltone \mid \contexttwo$}
      \DisplayProof
      &
      \AxiomC{$\pfd{\prooftwo}\pof\seq\contextone,\nltone,\contexttwo$}
      \UnaryInfC{$\seq\contextone,\nlttwo,\nltone,\contexttwo$}
      \DisplayProof
      \\ \vspace{-1pt} & \\
      \AxiomC{$\prooftwo\pof\contextone \seq \termone:\nltone \mid \contexttwo$}
      \RightLabel{$?w^\mu$}
      \UnaryInfC{$\contextone \seq \termone:\nltone \mid \contexttwo, \muvarone:\nlttwo $}
      \DisplayProof
      &
      \AxiomC{$\pfd{\prooftwo}\pof\seq\contextone,\nltone,\contexttwo$}
      \UnaryInfC{$\seq\contextone,\nltone,\nlttwo,\contexttwo$}     
      \DisplayProof
      \\ \vspace{-1pt} & \\
      \AxiomC{$\prooftwo\pof\contextone, \varone: \nltone, \vartwo: \nlttwo \seq \termone: \nltfour \mid \contexttwo$}
      \RightLabel{$?c^\lambda$}
      \UnaryInfC{$\contextone, \varthree: \nltthree \seq \subst{\subst{\termone}{\varone}{\varthree}}{\vartwo}{\varthree}: \nltfour \mid \contexttwo$}
      \DisplayProof
      &
      \AxiomC{$\pfd{\prooftwo}\pof\seq\contextone,\nltone,\nlttwo,\nltfour,\contexttwo$}
      \UnaryInfC{$\seq\contextone,\nltthree,\nltfour,\contexttwo$}     
      \DisplayProof
      \\ \vspace{-1pt} & \\
      \AxiomC{$\prooftwo\pof\contextone \seq \termone: \nltfour \mid \contexttwo, \muvarone: \nltone, \muvartwo: \nlttwo$}
      \RightLabel{$?c^\mu$}
      \UnaryInfC{$\contextone\seq \subst{\subst{\termone}{\muvarone}{\muvarthree}}{\muvartwo}{\muvarthree}: \nltfour \mid \contexttwo, \muvarthree: \nltthree $}
      \DisplayProof
      &
      \AxiomC{$\pfd{\prooftwo}\pof\seq\contextone,\nltfour,\nltone,\nlttwo,\contexttwo$}
      \UnaryInfC{$\seq\contextone,\nltfour,\nltthree,\contexttwo$} 
      \DisplayProof
      \\ \vspace{-1pt} & \\ \hline\hline
    \end{tabular}}
  \vspace{5pt}
\caption{Mapping of (multiplicative) derivations into $\BLLP$ proofs}
\label{fig:mapping}
\end{figure*}}
\shortv{by induction on the structure of $\proofone$, closely following Laurent's translation~\cite{laurent2003polarized}.}
This way one not only gets some guiding principles for subject-reduction, but can also prove that the underlying
transformation process is nothing more than cut-elimination:
\longv{
\begin{lemma}[$\lambda$-Substitution]\label{Lem:lamsub}
If $\proofone\pof\contextone, \varone: \DLABEL{\varthree}{\polfour}{\nfone}{\vartwo}{\polone}\seq \termone: \LABEL{\nftwo}{\vartwo}{\poltwo}\mid\contexttwo$ 
and $\prooftwo\pof\contextthree\seq \termtwo: \LABEL{\nfone}{\varthree}{\polfour}\mid\contextfour$, then for all $\polnine\resleq\poltwo$ there is
$\proofthree_\polnine$ such that
$$
\proofthree_\polnine\pof\contextone, \SUM_{\varnine <\polnine}\seq \subst{\termone}{\varone}{\termtwo}: \LABEL{\nftwo}{\vartwo}{\poltwo}\mid\contexttwo, \SUM_{\varnine< \polnine} \contextfour.
$$
Moreover, the proof obtained by $\polnine$-boxing $\pfd{\prooftwo}$ and cutting it against $\pfd{\proofone}$ is guaranteed to $\redcutelimarr$-reduce
to $\proofthree_\polnine$. 
\end{lemma}
\begin{proof}
As usual, this is an induction on the structure of $\proofone$. We only need to be careful and generalize the statement to
the case in which a \emph{simultaneous} substitution for many variables is needed.
\end{proof}
\begin{lemma}[$\mu$-Substitution]\label{Lem:musub}
If $\proofone\pof\contextone \seq \termone: \LABEL{\bot}{\vartwo}{\poltwo}\mid\contexttwo, \muvarone: \LABEL{\nfone\arrow{\varthree}{\polfour} \nftwo}{\vartwo}{\polone}$ and 
$\prooftwo\pof \contextthree \seq \termtwo: \LABEL{\nfone}{\varthree}{\polfour}\mid\contextfour$, then for all $\polnine\resgeq\poltwo$ there is 
$\proofthree_\polnine$ such that
$$
\contextone,\SUM_{\varnine < \polnine}\contextthree\seq\subst{\termone}{[\muvarone]\termfour}{[\muvarone](\termfour)\termtwo}: \LABEL{\bot}{\vartwo}{\poltwo}\mid\contexttwo, \muvarone: \LABEL{\nftwo}{\vartwo}{\polone}, \SUM_{\varnine < \polnine} \contextfour
$$
Moreover, the proof obtained by $\polnine$-boxing $\pfd{\prooftwo}$, tensoring it with an axiom and cutting the result against 
$\pfd{\proofone}$ is guaranteed to $\redcutelimarr$-reduce to $\proofthree_\polnine$.
\end{lemma}
}
\begin{theorem}[Subject Reduction]\label{theo:subjred}
Let $\proofone\pof\tj{\contextone}{\termone}{\nltone}{\contexttwo}$ and suppose
$\termone\toh\termtwo$. Then there is $\prooftwo\pof\tj{\contextone}{\termtwo}{\nltone}{\contexttwo}$. Moreover
$\pfd{\proofone}\redcutelimarr^+\pfd{\prooftwo}$.
\end{theorem}
\longv{
\begin{proof}
By induction on the structure of $\proofone$.
Here are some interesting cases:
\begin{varitemize}
\item 
  If $\termone$ is an application, reduction takes place inside $\termone$, and $\proofone$
  is as follows
  \begin{prooftree}
    \AxiomC{$\contextone \seq \termone : \LABEL{\nfone \arrow{\varone}{\polone} \nftwo}{\vartwo}{\poltwo} | \contexttwo$}
    \AxiomC{$\contextthree \seq \termthree : \LABEL{\nfone}{\vartwo}{\polone} | \contextfour$}
    \AxiomC{$\polnine \resgeq \poltwo, \polten \resgeq \poltwo$}
    \RightLabel{app}
    \TrinaryInfC{$\contextone \contsum \SUM_{\varnine<\polnine} \contextthree \seq (\termone)\termthree: \LABEL{\nftwo}{\varthree}{\polten} | \contexttwo \contsum \SUM_{\varnine<\polnine}\contextfour$}
  \end{prooftree}
  then $\prooftwo$ is
  \begin{prooftree}
    \AxiomC{$\contextone \seq \termtwo : \LABEL{\nfone \arrow{\varone}{\polone} \nftwo}{\vartwo}{\poltwo} | \contexttwo$}
    \AxiomC{$\contextthree \seq \termthree : \LABEL{\nfone}{\vartwo}{\polone} | \contextfour$}
    \AxiomC{$\polnine \resgeq \poltwo, \polten \resgeq \poltwo$}
    \RightLabel{app}
    \TrinaryInfC{$\contextone \contsum \SUM_{\varnine<\polnine} \contextthree \seq (\termtwo)\termthree: \LABEL{\nftwo}{\varthree}{\polten} | \contexttwo \contsum \SUM_{\varnine<\polnine}\contextfour$}
  \end{prooftree}
  which exists by induction hypothesis. We omit the other trivial cases.
\item 
  If $\termone$ is a $\beta$-redex, then $\proofone$ looks as follows:
  \begin{prooftree}
    \AxiomC{$\contextone, \varone : \DLABEL{\varthree}{\polfour}{\nfone}{\vartwo}{\polone} \seq \termone: \LABEL{\nftwo}{\vartwo}{\poltwo} | \contexttwo$}
    \AxiomC{$\polthree\resgeq \poltwo, \polthree\resgeq \polone$}
    \RightLabel{abs}
    \BinaryInfC{$\contextone \seq \lambda \varone.\termone: \LABEL{\nfone \arrow{\vartwo}{\polfour} \nftwo}{\varfour}{\polthree} | \contexttwo$}
    \AxiomC{$\contextthree \seq \termtwo : \LABEL{\nfone}{\varthree}{\polfour} | \contextfour$}
    \AxiomC{$\polnine \resgeq \poltwo, \polten \resgeq \poltwo$}
    \RightLabel{app}
    \TrinaryInfC{$\contextone,\SUM_{\varnine<\polnine} \contextthree \seq (\lambda \varone.\termone)\termtwo: \LABEL{\nftwo}{\vartwo}{\polten} | \contexttwo,\SUM_{\varnine<\polnine}\contextfour$}
  \end{prooftree}
  Lemma~\ref{Lem:lamsub} ensures that the required type derivation actually exists:
  \begin{prooftree}
    \AxiomC{$\contextone,\SUM_{\varnine<\polnine} \contextthree \seq \subst{\termone}{\varone}{\termtwo}: \LABEL{\nftwo}{\vartwo}{\polten} | \contexttwo,\SUM_{\varnine<\polnine}\contextfour$}
  \end{prooftree}
\item 
  If $\termone$ is a $\mu$-redex, then $\proofone$ looks as follows:
  \begin{prooftree}
    \AxiomC{$\contextone \seq \termone: \LABEL{\bot}{\varthree}{\poltwo} | \muvartwo: \LABEL{\nfone \arrow{\varthree}{\polfour} \nftwo}{\vartwo}{\polone}, \contexttwo$}
    \RightLabel{$\mu$-abs}
    \UnaryInfC{$\contextone \seq \mu \muvartwo \termone: \LABEL{\nfone \arrow{\varthree}{\polfour} \nftwo}{\vartwo}{\polone} | \contexttwo$}
    \AxiomC{$\contextthree \seq \termtwo : \LABEL{\nfone}{\vartwo}{\polfour} | \contextfour$}
    \AxiomC{$\polnine \resgeq \polone, \polten \resgeq \polone$}
    \RightLabel{app}
    \TrinaryInfC{$\contextone, \SUM_{\varnine<\polnine} \contextthree \seq (\mu \muvartwo.\termone)\termtwo: \LABEL{\nftwo}{\varthree}{\polten} | \contexttwo,\SUM_{\varnine<\polnine}\contextfour$}
  \end{prooftree}
  and Lemma~\ref{Lem:musub} ensures us that $\prooftwo$ exists for
  \begin{prooftree}
    \AxiomC{$\contextone \contsum \SUM_{\varnine<\polnine} \contextthree \seq \mu \muvartwo.\lsubst{\termone}{[\muvartwo]\termthree}{[\muvartwo](\termthree)\termtwo}: \LABEL{\nftwo}{\varthree}{\polten} | \contexttwo \contsum \SUM_{\varnine<\polnine}\contextfour$}
  \end{prooftree}
\item
  If $\termone$ is a $\theta$-redex, then $\proofone$ looks as follows:
  \begin{prooftree}
    \AxiomC{$\proofone \pof \contextone \seq \termone: \LABEL{\nfone}{\varone}{\polone} \mid \contexttwo$}
    \RightLabel{$?w^\mu$}
    \UnaryInfC{$\contextone \seq \termone: \LABEL{\nfone}{\varone}{\polone} \mid \contexttwo, \muvarone: \LABEL{\nfone}{\vartwo}{\poltwo}$}
    \AxiomC{$\polthree \resgeq \polone+\poltwo$}
    \RightLabel{\textsf{$\mu$-name}}
    \BinaryInfC{$\contextone \seq [\muvarone]\termone: \LABEL{\bot}{\varone}{\polfour} \mid \contexttwo, \muvarone: \LABEL{\nfone}{\vartwo}{\polthree}$}
    \RightLabel{\textsf{$\mu$-abs}}
    \UnaryInfC{$\contextone \seq \termone: \LABEL{\nfone}{\varone}{\polthree} \mid \contexttwo$}
  \end{prooftree}
  Since $\polthree=\polone+\poltwo \resgeq \polone$ we know that
  \begin{prooftree}
    \AxiomC{$\proofone^S \pof \contextone \seq \termone: \LABEL{\nfone}{\varone}{\polthree} \mid \contexttwo$}
  \end{prooftree}
  where $\proofone^S$ is the derivation obtained from $\proofone$, applying the Subtyping Lemma to the derivation $\proofone$.
\end{varitemize}
This concludes the proof.
\end{proof}
}
Observe how performing head reduction corresponds to \remove{following }$\redcutelimarr$, instead of
the more permissive $\cutelimarr$. The following, then, is an easy corollary of Theorem~\ref{theo:subjred} and Theorem~\ref{theo:polystepbllp}:
\begin{theorem}[Polystep Soundness for Terms]\label{theo:polystepterms}
Let $\proofone\pof\tj{\contextone}{\termone}{\nltone}{\contexttwo}$ and let $\termone\toh^n\termtwo$. 
Then $n\leq\polone_{\pfd{\proofone}}$. 
\end{theorem}
\section{Control Operators}
In this section, we show that $\BLLPLM$ is powerful enough to type (the natural encoding of) two popular control operators, namely
\Scheme's \callcc\ and Felleisen's {\FellC} \cite{ariola2003minimal} \cite{laurent2003polarized}.

Control operators change the evaluation context of an expression. This is simulated by the operators 
$\mu$ and $[\cdot]$ which can, respectively, save and restore a stack of arguments to be passed to subterms. 
This idea, by the way, is the starting point of an extension of Krivine's machine for de Groote's 
$\lambda\mu$ \cite{de1998environment} (see Section~\ref{sec:absmac}).
\longv{

An extension of de Groote's calculus named $\Lambda\mu$-calculus \cite{saurin2005separation} satisfies a B\"ohm separation 
theorem that fails for Parigot's calculus \cite{david2001calculus}. Hence in an untyped setting the original $\lambda\mu$ of Parigot 
is strictly less expressive than de Groote's calculus. 
}
\subsection{\large \callcc}
An encoding of \callcc\ into the $\lambda\mu$-calculus could be, e.g., 
$\kappa=\lambda \varone.\mu\muvarone.[\muvarone](\varone)\lambda\vartwo.\mu\muvartwo.[\muvarone]\vartwo$. 
Does $\kappa$ have the operational behavior we would expect from \callcc? First of all, it should satisfy the following property 
(see \cite{felleisen1990expressive}). If 
$k\not\in\FV{e}$, 
then $(\kappa)\lambda k.e\to^* e$. Indeed:
$$
(\lambda \varone.\mu\muvarone.[\muvarone](\varone)\lambda\vartwo.\mu\muvartwo.[\muvarone]\vartwo)\lambda k.e\toh\mu\muvarone.[\muvarone](\lambda k.e)\lambda\vartwo.\mu\muvartwo.[\muvarone]\vartwo
\toh \mu\muvarone.[\muvarone]e\toh e,
$$
where the second $\beta$-reduction step replaces $\subst{e}{k}{\lambda\vartwo.\mu\muvartwo.[\muvarone]\vartwo}$ with 
$e$ since $k\not\in \FV{e}$ by hypothesis. It is important to observe that the second step replaces a variable for 
a term with a free $\mu$-variable, hence weak reduction gets stuck. (Actually, our notion of 
weak reduction is even more restrictive than the one proposed by de Groote in \cite{de1998environment}.)
Head reduction, on the contrary, is somehow more liberal. Moreover, it is also straightforward to check that the 
reduction of {\callcc} in \cite[\S 3.4]{Parigot} can be simulated by head reduction on $\kappa$.

But is $\kappa$ typable in $\BLLPLM$? The answer is positive: a derivation typing it with (an instance of) Pierce's law is
in Figure~\ref{fig:callcc}, where $\proofone$ is the obvious derivation of
$$
\varone: \DLABEL{}{\polthree}{(\COATOMONE \arrow{}{\polfour} \COATOMTWO)\arrow{}{1}\COATOMONE}{}{1} \seq 
\varone: \LABEL{(\COATOMONE \arrow{}{\polfour} \COATOMTWO)\arrow{}{1} \COATOMONE}{}{\polthree}\mid\muvarone: \LABEL{\COATOMONE}{}{0}.
$$ 
\begin{figure*}
\fbox{
\scriptsize
\begin{minipage}{\textwidth}
\centering
\vspace{5pt}
\begin{prooftree}
\AxiomC{$\proofone$}
\AxiomC{}
\RightLabel{\textsf{var}}
\UnaryInfC{$\vartwo: \DLABEL{}{\polfour}{\COATOMONE}{}{1} \seq \vartwo: \LABEL{\COATOMONE}{}{\polfour}\mid\muvarone: \LABEL{\COATOMONE}{}{0}, \muvartwo: \LABEL{\COATOMTWO}{}{0}$}
\RightLabel{$\mu$\textsf{-name}}
\UnaryInfC{$\vartwo: \DLABEL{}{\polfour}{\COATOMONE}{}{1} \seq [\muvarone]\vartwo: \LABEL{\bot}{}{0}\mid\muvarone: \LABEL{\COATOMONE}{}{\polfour}, \muvartwo: \LABEL{\COATOMTWO}{}{0}$}
\RightLabel{$\mu$\textsf{-abs}}
\UnaryInfC{$\vartwo: \DLABEL{}{\polfour}{\COATOMONE}{}{1} \seq \mu\muvartwo.[\muvarone]\vartwo: \LABEL{\COATOMTWO}{}{0}\mid\muvarone: \LABEL{\COATOMONE}{}{\polfour}$}
\RightLabel{\textsf{abs}}
\UnaryInfC{$\seq \lambda\vartwo.\mu\muvartwo.[\muvarone]\vartwo: \LABEL{\COATOMONE\arrow{}{\polfour} \ATOMTWO}{}{1}|\muvarone: \LABEL{\COATOMONE}{}{\polfour}$} 
\RightLabel{\textsf{app}}
\BinaryInfC{$\varone: \DLABEL{\varfive}{\polthree}{(\COATOMONE\arrow{}{\polfour} \COATOMTWO)\arrow{}{1} \COATOMONE}{}{1} \seq (\varone)\lambda\vartwo.\mu\muvartwo.[\muvarone]\vartwo: 
  \LABEL{\COATOMONE}{\varfive}{\polthree}|\muvarone: \LABEL{\COATOMONE}{\varfive}{\sum_{\varfive<\polthree} \polfour}$}
\AxiomC{$\begin{array}{rcl}\polten&\resgeq&\polthree+\sum_{\varfive<\polthree}\polfour \\ \polten &\resgeq& 1\end{array}$}
\RightLabel{$\mu$\textsf{-name}}
\BinaryInfC{$\varone: \DLABEL{\varfive}{\polthree}{(\COATOMONE\arrow{}{\polfour} \COATOMTWO)\arrow{}{1} \COATOMONE}{}{0} \seq 
  [\muvarone](\varone)\lambda\vartwo.\mu\muvartwo.[\muvarone]\vartwo: \LABEL{\bot}{}{1}\mid\muvarone: \LABEL{\COATOMONE}{}{\polten}$}
\RightLabel{$\mu$\textsf{-abs}}
\UnaryInfC{$\varone: \DLABEL{\varfive}{\polthree}{(\COATOMONE\arrow{\varfive}{\polfour} \COATOMTWO)\arrow{}{1} \COATOMONE}{}{1} \seq 
  \mu\muvarone.[\muvarone](\varone)\lambda\vartwo.\mu\muvartwo.[\muvarone]\vartwo: \LABEL{\COATOMONE}{}{\polten}\mid$}
\RightLabel{\textsf{abs}}
\UnaryInfC{$\seq \lambda \varone.\mu\muvarone.[\muvarone](\varone)\lambda\vartwo.\mu\muvartwo.[\muvarone]\vartwo: \LABEL{((\COATOMONE\arrow{}{\polfour} \COATOMTWO)\arrow{}{1} 
  \COATOMONE)\arrow{\varfive}{\polthree} \COATOMONE}{}{\polten}\mid$}
\end{prooftree}\
\end{minipage}}
\caption{A Type Derivation for $\kappa$}
\label{fig:callcc}
\end{figure*}

\subsection{Felleisen's {\large \FellC}}
The canonical way to encode Felleisen's \FellC\ as a $\lambda\mu$-term is as the term $\aleph=\lambda\varten.\mu\muvarone.(\varten)\lambda \varone.[\muvarone]\varone$.
Its behavior should be something like $(\aleph)\termfour t_1 \dots t_k\to (\termfour)\lambda x.(x)t_1 \dots t_k$, where $x\not\in FV(t_1, \dots, \linebreak[1] t_k)$, 
i.e., $x$ is a fresh variable. Indeed
$$
(\aleph)\termfour t_1 \dots t_k\toh(\mu\alpha\remove{:\nffour}.(\termfour)\lambda x.[\alpha](x))t_1 \dots t_k\toh^k
\mu\alpha\remove{:\nfone}.(\termfour)\lambda x\remove{:(\nfone_1\arrow{\vartwo_1}{\polone_1} \dots \nfone_k\arrow{\vartwo_k}{\polone_k} \nfone)}.[\alpha](x)t_1\dots t_k.
$$
A type derivation for $\aleph$ is in Figure~\ref{fig:fellc}, where $\proofthree$ is a derivation for
$$
\varten:\DLABEL{\varfive}{\polnine}{\pneg{}{1}{\pneg{}{\polthree}{\COATOMONE}}}{}{1} \seq 
\varten:\LABEL{\pneg{}{1}{\pneg{}{\polthree}{\COATOMONE}}}{\varfive}{\polnine}\mid\muvarone:\LABEL{\COATOMONE}{}{0}.
$$ 
\begin{figure*}[bpt]
\fbox{
\scriptsize
\begin{minipage}{0.98\textwidth}
\centering
\vspace{5pt}
\begin{prooftree}
\AxiomC{$\proofthree$}
\AxiomC{}
\RightLabel{\textsf{var}}
\UnaryInfC{$\varone: \DLABEL{}{\polthree}{\COATOMONE}{}{1} \seq \varone: \LABEL{\COATOMONE}{}{\polthree}$}
\RightLabel{$\mu$\textsf{-name}}
\UnaryInfC{$\varone: \DLABEL{}{\polthree}{\COATOMONE}{}{1} \seq [\muvarone]\varone: \LABEL{\bot}{}{0} \mid \muvarone: \LABEL{\COATOMONE}{}{\polthree}$}
\RightLabel{\textsf{abs}}
\UnaryInfC{$\seq \lambda \varone.[\muvarone]\varone: \LABEL{\pneg{}{\polthree}{\COATOMONE}}{}{1} \mid \muvarone: \LABEL{\COATOMONE}{}{\polthree}$}
\RightLabel{\textsf{app}}
\BinaryInfC{$\varten: \DLABEL{\varfive}{\polnine}{\pneg{}{1}{\pneg{}{\polthree}{\COATOMONE}}}{}{1} \seq (\varten)\lambda \varone.[\muvarone]\varone: \LABEL{\bot}{\varfive}{\polnine} 
   \mid \muvarone: \LABEL{\COATOMONE}{}{\sumlf{\varfive}{\polnine}{\polthree}}$}
\RightLabel{$\mu$\textsf{-abs}}
\UnaryInfC{$\varten: \DLABEL{\varfive}{\polnine}{\pneg{}{1}{\pneg{}{\polthree}{\COATOMONE}}}{}{1} \seq \mu\muvarone.(\varten)\lambda \varone.[\muvarone]\varone: \LABEL{\COATOMONE}{}{\sumlf{\varfive}{\polnine}{\polthree}}\mid$}
\AxiomC{$\begin{array}{rcl}\polten&\resgeq& 1\\\polten&\resgeq&\sumlf{\varfive}{\polnine}{\polthree}\end{array}$}
\RightLabel{\textsf{abs}}
\BinaryInfC{$\seq \lambda\varten.\mu\muvarone.(\varten)\lambda \varone.[\muvarone]\varone: 
   \LABEL{\pneg{}{1}{\pneg{}{\polthree}{\COATOMONE}}\arrow{\varfive}{\polnine}\COATOMONE}{}{\polten}$}
\end{prooftree}\ 
\end{minipage}
}
\caption{A Type Derivation for $\aleph$}
\label{fig:fellc}
\end{figure*}
It is worth noting that weak reduction is strong enough to properly simulating the operational behavior of $\FellC$.
It is not possible to type $\mathcal{C}$ in Parigot's $\lambda\mu$, unless an open term is used. Alternatively, a 
free continuation constant must be used (obtaining yet another calculus \cite{ariola2003minimal}). 
This is one of the reasons why we picked the version of $\lambda\mu$-calculus proposed by de Groote over other calculi. 
See \cite{de1994relation} for a discussion about $\lambda\mu$-and-$\lambda$-calculi and Felleisen's $\mathcal{C}$.

\section{Abstract Machines}\label{sec:absmac}
Theorem~\ref{theo:polystepterms}, the main result of this paper so far, tells us that the number 
of \emph{head-reduction steps} performed by terms typable in $\BLLPLM$ is bounded 
by the weight of the underlying type derivation. One may wonder, however, whether taking the 
number of reduction steps as a measure of term complexity is sensible or not --- substitutions
involve arguments which can possibly be much bigger than the original term. Recent work by Accattoli and the first 
author~\cite{accattoli2012invariance}, however, shows that in the case of $\lambda$-calculus endowed with head reduction, 
the unitary cost model is polynomially invariant with respect to Turing machines. We conjecture that 
those invariance results can be extended to the $\lambda\mu$-calculus.

\shortv{It can be shown that $\BLLPLM$ is polystep sound for another cost model, namely the one
induced by de Groote's $\K$, an abstract machine for the $\lambda\mu$-calculus. This is done
following a similar proof for $\PCF$ typed with linear dependent types~\cite{DalLago2011} and Krivine's Abstract
Machine (of which $\K$ is a natural extension). The main idea consists in extending $\BLLP$ to a type system for $\K$'s configurations,
this way defining a \emph{weight} for each of them in the form of a resource polynomial. The weight, as expected, can
then be shown to decrease at each $\K$'s computation step. It is worth noting that the weight defined this way is
fundamentally different than the one from Section~\ref{subsec:soundness}. See \cite{EV} for some more details.}
\longv{In this Section, we show that $\BLLPLM$ is polystep sound for another cost model, namely the one
induced by de Groote's $\K$, an abstract machine for the $\lambda\mu$-calculus. This will be done
following a similar proof for $\PCF$ typed with linear dependent types~\cite{DalLago2011} and Krivine's Abstract
Machine (of which $\K$ is a natural extension).
}

\longv{
Configurations of $\K$ are built around environments, closures and stacks, which are defined mutually
recursively as follows:
\begin{varitemize}
\item
  \emph{Environments} are partial functions which makes $\lambda$-variables correspond
  to closures and $\mu$-variables correspond to stacks; metavariables for
  environments are $\envone,\envtwo$, etc.;
\item
  \emph{Closures} are pairs whose first component is a $\lambda\mu$-term and whose second component
  is an environment; metavariables for closure are $\closone,\clostwo$, etc.
\item
  \emph{Stacks} are just finite sequences of closures; metavariables for stacks are $\stone,\sttwo$, etc.
\end{varitemize}
Configurations are pairs whose first component is a closure and whose second component is a stack, and are
indicated with $\confone,\conftwo$, etc.
Reduction rules for configurations are in Figure~\ref{fig:kmachine}. 
\begin{figure*}
  \begin{center}
  \fbox{
  \begin{minipage}{.98\textwidth}
    \begin{align*}
      ((\varone,\envone),\stone)&\toconf(\envone(\varone),\stone);\\
      ((\lambda\varone.\termone,\envone),\closone\cdot\stone)
        &\toconf((\termone,\subst{\envone}{\varone}{\closone}),\stone);\\
      ((\termone\termtwo,\envone),\stone)&\toconf
        ((\termone,\envone),(\termtwo,\envone)\cdot\stone);\\
      ((\mu\alpha.\termone,\envone),\stone)&\toconf
        ((\termone,\subst{\envone}{\alpha}{\stone}),\varepsilon);\\
      (([\alpha]\termone,\envone),\varepsilon)&\toconf
        ((\termone,\envone),\envone(\alpha)).\\
    \end{align*}
  \end{minipage}}
  \end{center}
  \caption{$\K$-machine Transitions.}\label{fig:kmachine}
\end{figure*}
The $\K$-machine is sound and complete with respect to head reduction~\cite{de1998environment}, where however, reduction 
can take place in the scope of $\mu$-abstractions, but not in the scope of $\lambda$-abstractions.\longv{\footnote{The authors are aware of the work in \cite{streicher1998classical}, in which a Krivine machine for $\lambda\mu$ 
is derived semantically rather than syntactically (independently of de Groote). In the same paper there is also a further extension of the machine which allows to reduce under $\mu$- and even $\lambda$-abstractions. The paper is not essential for our purposes since the abstract machine of de Groote is enough to work with control operators. Still, even though there are some important differences with respect to our setting (the calculus considered is an untyped variant of Parigot's $\lambda\mu$), it might be worthwhile to investigate in the future.}}

Actually, $\BLLPLM$ can be turned into a type system for $\K$'s \emph{configurations}.  
\longv{ We closely follow Laurent~\cite{Laurent03note} here. 
}
The next step is to assign a weight $\poltwo^\proofone$ to every type derivation
$\proofone\pof\confone:\nltone$, similarly to what we have done in type derivations for \emph{terms}.
The idea then is to prove that the weight of (typable) configurations decreases at every transition
step:
\begin{lemma}
If $\confone\toconf\conftwo$, then $\poltwo^\confone\resgt\poltwo^\conftwo$.
\end{lemma}
This allows to generalize polystep soundness to $\K$:
\begin{theorem}[Polystep Soundness for the $\K$]\label{theo:polystepmachine}
Let $\proofone\pof\tjc{\confone}{\nltone}$ and let $\confone\toconf^n\conftwo$. Then $n\leq\poltwo^{\proofone}$. 
\end{theorem}
Please observe how Theorem~\ref{theo:polystepmachine} holds in particular when
$\confone$ is the initial configuration for a typable term $\termone$, i.e., 
$\pair{\pair{\termone}{\emptyenv}}{\emptystk}$.}
\section{Conclusions}
In this paper\remove{, some evidence has been given on the fact} we have presented some evidence that the enrichment to intuitionistic linear logic provided
by bounded linear logic is robust enough to be lifted to polarized linear logic and the $\lambda\mu$-calculus. This
paves the way towards a complexity-sensitive type system, which on the one hand guarantees that typable terms
can be reduced to their normal forms in a number of reduction steps which can be read from their type derivation, and
on the other allows to naturally type useful control operators.

Many questions have been purposely left open here: in particular, the language of programs is the pure, constant-free, $\lambda\mu$-calculus,
whereas the structure of types is minimal, not allowing any form of polymorphism. We expect that endowing $\BLLP$ with second
order quantification or $\BLLPLM$ with constants and recursion should not be particularly problematic, although laborious: the
same extensions have already been considered in similar settings in the absence of control~\cite{GSS92TCS,DalLago2011}. Actually,
a particularly interesting direction would be to turn $\BLLPLM$ into a type system for Ong and Stewart's $\muPCF$~\cite{OngS97POPL}, this
way extending the linear dependent paradigm to a language with control. This is of course outside the scope of this paper, whose purpose
was only to delineate the basic ingredients of the logic and the underlying type system.

As we stressed in the introduction, we are convinced this work to be the first one giving a time complexity analysis methodology
for a programming language with higher-order functions \emph{and control}.\longv{\footnote{Tatsuta has investigated the maximum length of $\mu$-reduction 
for a language \emph{without $\lambda$-abstractions} (RTA 2007).}} One could of course object that complexity analysis of
$\lambda\mu$-terms could be performed by translating them into equivalent $\lambda$ terms, e.g. by way of a suitable CPS-transform~\cite{DBLP:conf/caap/Groote94}.
This, however, would force the programmer (or whomever doing complexity analysis) to deal with programs which are structurally
different from the original one. And of course, translations could introduce inefficiencies, which are maybe harmless from a purely
qualitative viewpoint, but which could make a difference for complexity analysis.\remove{\longv{\footnote{CPS 
translations in general complicate types, e.g., the CPS translation in \cite{DBLP:conf/caap/Groote94} corresponds to the 
negative translation. Thus existing techniques for complexity analysis on $\lambda$ calculus are harder to apply.}}}
\longv{\bibliographystyle{plain}}
\shortv{\bibliographystyle{abbrv}}
\bibliography{lpar2013}
\end{document}